\documentclass[11pt,a4paper]{article}
\usepackage{vmargin}
\setmarginsrb{1.0in}{1.0in}{1.0in}{1.0in}{0mm}{0mm}{5mm}{5mm}

\usepackage[utf8]{inputenc}
\usepackage[T1]{fontenc}
\usepackage{rotating, graphicx}
\title{Fair Tree Connection Games with\\ Topology-Dependent Edge Cost\\{\small (full version)}}
\author{Davide Bil{\`{o}}\thanks{Department of Humanities and Social Sciences, University of Sassari, Italy, \texttt{davide.bilo@uniss.it}} \and Tobias Friedrich\thanks{Hasso Plattner Institute, University of Potsdam, Germany, \texttt{firstname.lastname@hpi.de}} \and Pascal Lenzner\footnotemark[2] \and  Anna Melnichenko\footnotemark[2] \and Louise Molitor\footnotemark[2]}

\usepackage{tabularx} %for the tables
\usepackage{amssymb,amsmath,amsthm}
\usepackage{multirow}
\usepackage{booktabs} %for tables
\usepackage{graphicx}
\usepackage{subcaption}
\usepackage{mathtools}
\usepackage{color}
\usepackage{thmtools}
\usepackage{hyperref}
\usepackage{cleveref}

\usepackage[ruled,vlined,linesnumbered,algo2e]{algorithm2e}
\SetKwFor{Function}{function}{is}{end function}
\SetKw{Input}{input}
\SetKw{Output}{output}
\SetKw{Initialization}{initialization}

\newcommand{\SNCG}{TCG\xspace}
\newcommand{\OPT}{OPT\xspace} %social optimum

\newcommand{\cost}{{\mathit cost}}
\newcommand{\indeg}{\mathit indeg}
\newcommand{\dd}{{\mathit d}}
\newcommand{\stp}{\mathbf s} %strategy profile
\newcommand{\FR}{FR\xspace} %fairness ratio

\newtheorem{remark}{Remark}
\newtheorem{theorem}[remark]{Theorem}
\newtheorem{lemma}[remark]{Lemma}
\newtheorem{corollary}[remark]{Corollary}
\newtheorem{conjecture}{Conjecture}

\date{~}

\begin{document}

\maketitle
\begin{abstract}
	\noindent How do rational agents self-organize when trying to connect to a common target? We study this question with a simple tree formation game which is related to the well-known fair single-source connection game by Anshelevich et al.~(FOCS'04) and selfish spanning tree games by Gourvès and Monnot (WINE'08). In our game agents correspond to nodes in a network that activate a single outgoing edge to connect to the common target node (possibly via other nodes). Agents pay for their path to the common target, and edge costs are shared fairly among all agents using an edge. The main novelty of our model is dynamic edge costs that depend on the in-degree of the respective endpoint. This reflects that connecting to popular nodes that have increased internal coordination costs is more expensive since they can charge higher prices for their routing service.
	
	In contrast to related models, we show that equilibria are not guaranteed to exist, but we prove the existence for infinitely many numbers of agents. Moreover, we analyze the  structure of equilibrium trees and employ these insights to prove a constant upper bound on the Price of Anarchy as well as non-trivial lower bounds on both the Price of Anarchy and the Price of Stability. We also show that in comparison with the social optimum tree the overall cost of an equilibrium tree is more fairly shared among the agents. 
	Thus, we prove that self-organization of rational agents yields on average only slightly higher cost per agent compared to the centralized optimum, and at the same time, it induces a more fair cost distribution. Moreover, equilibrium trees achieve a beneficial trade-off between a low height and low maximum degree, and hence these trees might be of independent interest from a combinatorics point-of-view. We conclude with a discussion of promising extensions of our model.    
\end{abstract}

\section{Introduction}
Network Design is an important optimization problem where for a given weighted host graph and a given set of terminal pairs the cheapest subgraph which connects all terminal pairs has to be found. Besides an abundance of research works with an optimization point-of-view, e.g. see the survey by Magnanti and Wong~\cite{MW84}, a strategic version of the Network Design problem~\cite{ADTW08,ADKTWR} has kindled significant interest in recent years. In the \emph{connection game}, a weighted host graph $H$ is given and $n$ agents with given terminal node pairs $(s_i,t_i)$, for $1\leq i \leq n$, strategically select $s_i$-$t_i$-paths in~$H$ to connect their respective terminal nodes. The union of the selected paths forms a subgraph $G$ of~$H$ which constitutes the actually designed network.  
The usage cost of each edge of $H$ corresponds to its weight, and agents using some edge $e$ in $H$ have to pay this cost. If an edge $e$ is used by more than one agent, then a cost-sharing protocol determines how the usage cost of $e$ is split among its users. One of the most common cost-sharing protocols is Shapley cost-sharing where each agent pays a fair share of the edge cost, i.e., the cost-share is the edge cost divided by the number of users. This game-theoretic setting, called \emph{fair connection game}, was investigated by Anshelevich et al.~\cite{ADKTWR} and has since become an influential paper in Algorithmic Game Theory. An important special case is the setting in which all the strategic agents want to connect to a common source node. This variant, where $t_1 = \dots = t_n$ and where every other node is a terminal node of some agent, is usually denoted as the \emph{(fair) single-source connection game}, with the interpretation that all the agents want to connect to a common source node to receive broadcast messages and that the edge cost for connecting to the common source is paid by the downstream users.

A similar related game-theoretic setting are \emph{selfish spanning tree games}~\cite{GM08}. There a weighted complete host graph with $n+1$ nodes, consisting of a common target node $r$ and $n$ nodes which correspond to selfish agents, is given and every selfish agents now selects an incident edge to connect to the common target node $r$ either directly or indirectly via selected edges of other agents. The cost of an agent is then determined by its unique path to $r$. Thus, in any equilibrium the subgraph of all selected edges forms a spanning tree rooted at $r$.   

This paper sets out to investigate a game-theoretic Network Design model that is closely related to the fair single-source connection game and to selfish spanning tree games. The main novel feature of our model is the twist that the cost of the edges in the formed spanning tree depend on its topology. In particular, we consider dynamic edge costs which are proportional to the in-degree of the node they connect to. Network nodes with high in-degree can be considered as popular, and we assume that connecting to popular nodes is more expensive than connecting to unpopular nodes. These dynamic edge costs can also be understood as the internal cost of a node for coordinating data traffic coming from different connections. A node with many incoming edges and thus higher internal coordination cost can charge higher prices for serving each of the incoming edges.   

To the best of our knowledge, we define and analyze the first (game-theoretic) Network Design model where the edge costs depend on the topology of the formed network. We believe that this model sheds light on settings where the actual charges for establishing links are determined by supply and demand and the agents act strategically to optimize their cost for receiving their desired service. 

\subsection{Model, Definition, Notation}
We consider a strategic game called \emph{fair tree connection game with topology-dependent edge cost}, or \emph{tree connection game (\SNCG)} for short. In the \SNCG we will consider a given unweighted complete directed host graph $H = (V,E)$, where $V$ is the set of nodes and $E$ is the set of edges of~$H$. The host graph $H$ consists of $n+1$ nodes $V = \{r,v_1,\dots,v_n\}$ where node $r$ is the common target node, also called the root, and every node $v_i$, for $1\leq i \leq n$, corresponds to a selfish agent $i$ striving to be connected to the root $r$. For this, every agent $i$ strategically activates a single incident edge $(v_i,s_i)$, where $s_i\in V\setminus \{v_i\}$. Hence, the strategy space of each agent is the set of other nodes to connect to. Given a strategy profile $\stp = (s_1,\dots,s_n)$, i.e., an $n$-dimensional vector where the $j$-th entry corresponds to the node to which agent $j$ wants to activate her edge, we consider the directed network $T(\stp) = (V,E(\stp))$ which is induced by all the activated edges, i.e., $E(\stp) = \{(v_i,s_i) \mid 1\leq i \leq n\}$. We will see later that $T(\stp)$ is a spanning tree rooted at $r$ if $\stp$ is an equilibrium state of the \SNCG, hence the name. 

The cost of agent $i$ in the network $T(\stp)$ depends on its unique path $P_i$ in $T(\stp)$ to the root $r$ (if such a path exists). In case of existence, the path $P_i$ must be unique, since the out-degree of every node in $T(\stp)$ is at most $1$. More precisely, let $P_i$ be the directed path from $v_i$ to $r$ in $T(\stp)$, let $\indeg_{T(\stp)}(v)$ denote the number of edges with endpoint $v$ in $T(\stp)$, let $T(u)$ denote the subgraph of $T(\stp)$ rooted at node $u$, i.e., the subgraph of $T(\stp)$ induced by the nodes $u$ and every node which has a directed path to~$u$ and let $|T(u)|$ denote the number of nodes in $T(u)$. See Figure~\ref{fig:notation}. 
\begin{figure}[h]
 \centering
 \includegraphics[width=0.7\textwidth]{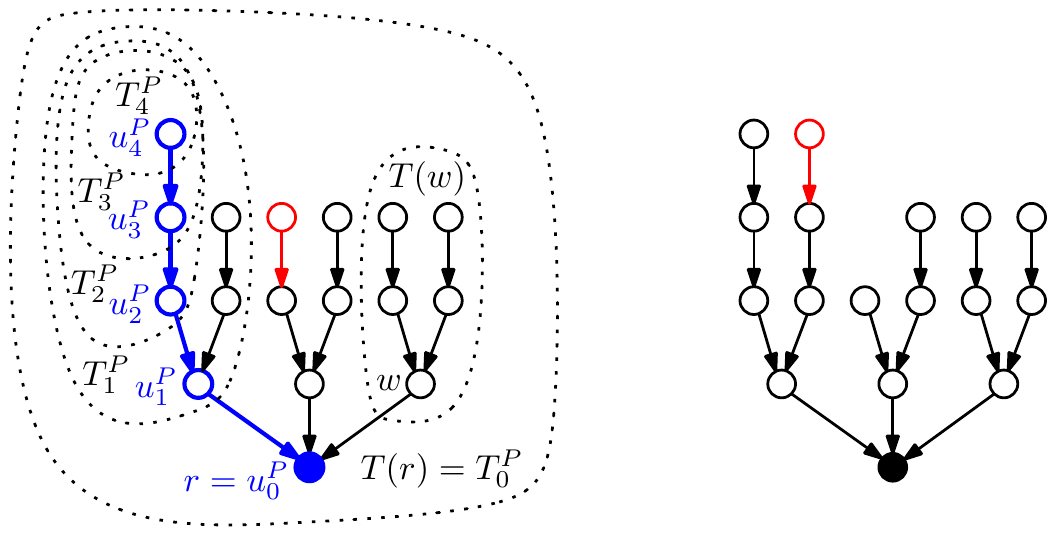}
 \caption{Left: $T(\stp)$ for $n=16$ agents. The path $P$ is colored blue and we have $d_0^P = 3, d_1^P = 2, d_2^P = d_3^P = 1, d_4^P = 0$ and $|T_1^P| = 6, |T_2^P| = 3, |T_3^P| = 2, |T_4^P| = 1, |T(w)| = 5$. Nodes $w$ and $u_1^P$ and also their corresponding agents are siblings.
 The shown network $T(\stp)$ is not stable since the agent colored red with cost $1 + \frac22 + \frac35 = \frac{13}{5}$ has an improving move. Right: the network after the agent colored red improved its cost to $1 + \frac12 + \frac23 + \frac37 = \frac{109}{42} < \frac{13}{5}$.}
 \label{fig:notation}
\end{figure}

\noindent The cost of agent $i$ in $T(\stp)$ then is 
$ \cost_{T(\stp)}(i) := \sum_{(u,v) \in P_i}\frac{\indeg_{T(\stp)}(v)}{|T(u)|}$, if $P_i$ exists and $\infty$ otherwise.
%\begin{equation*}
%	\cost_{T(\stp)}(i) := \begin{cases}
%		\sum_{(u,v) \in P_i}\frac{\indeg_{T(\stp)}(v)}{|T(u)|},	&	\text{if $P_i$ exists,}\\
%		\infty,													&	\text{otherwise.}
%	\end{cases}
%\end{equation*}
This cost function has the following very natural interpretation: the cost of activating edge $(u,v) $ from node $u$ to node $v$ is equal to node $v$'s in-degree, and this cost is fairly shared by all agents who use edge $(u,v)$ on their path towards the root $r$, i.e., by all agents in $T(u)$. We assume that each agent activates a single edge strategically to minimize its cost in the induced network $T(\stp)$. Clearly, since every agent $i$ can activate the edge $(v_i,r)$, $i$ can enforce finite cost by enforcing that the path $P_i$ exists. 

Consider a strategy profile $\stp = (s_1,\dots,s_{i-1},s_i,s_{i+1},\dots,s_n)$. We say that agent $i$ has \emph{an improving move in $\stp$} if $i$ has some alternative strategy $s_i' \neq s_i$ such that for the induced strategy profile $\stp' = (s_1,\dots,s_{i-1},s_i',s_{i+1},\dots,s_n)$ we have $cost_{T(\stp')}(i) < cost_{T(\stp)}(i)$, i.e., agent $i$ can strictly decrease its cost by activating a different outgoing edge. With this, we define the strategy profile $\stp$ to be in \emph{pure Nash equilibrium (NE)} or to be \emph{stable} if no agent has an improving move in $\stp$. If the context is clear, we use strategy profiles and their induced network interchangeably, i.e., we say that the network $T(\stp)$ is in NE or stable, if $\stp$ is in NE. Moreover, when we refer to some network $T(\stp)$ we will from now on omit the reference to the strategy profile $\stp$ and call the network simply $T$. Every stable network $T$ must be a spanning tree rooted at $r$, since every agent $i$ can activate the edge $(v_i,r)$ to achieve finite cost.   

The \emph{social cost} $SC(T)$ of a network $T$ is simply the sum over all agents' costs, i.e., 
$$SC(T)=\sum_{i=1}^n \cost_{T}(i) =\sum\limits_{v_i\in V}{\sum\limits_{(u,v)\in P_i}\frac{\indeg_{T}(v)}{|T(u)|}}
=\sum\limits_{(u,v)\in E}\frac{\indeg(v)}{|T(u)|}\cdot |T(u)|=\sum\limits_{v\in V }(\indeg(v))^2.$$ Note that $SC(T)$ nicely reflects the overall cost impact of the nodes' popularity or coordination costs which scales quadratically with the in-degree of a node. For a given number of agents $n$, let $OPT_n$ denote the network which minimizes the social cost. Moreover, if stable networks exist for~$n$ agents, we let $worstNE_n$ denote the stable network with the highest social cost and $bestNE_n$ the stable network with the lowest social cost. We define the \emph{Price of Anarchy (PoA)}~\cite{KP99} as $PoA = \sup_n \frac{SC(worstNE_n)}{OPT_n}$ and the \emph{Price of Stability (PoS)}~\cite{ADKTWR} as $PoS = \sup_n \frac{SC(bestNE_n)}{OPT_n}$, where the supremum is taken over all $n$ that admit a stable network. Besides the PoA and the PoS, that both focus on the overall cost and compare with the cost of a centrally designed social optimum network, we use a measure of the quality of networks which focuses on the cost distribution among the agents, called the \emph{Fairness Ratio (\FR)}, analogously to the \emph{utility uniformity} introduced in~\cite{feldman2008proportional}.
For a given network $T$, the $\FR(T)$ is the ratio between the maximum and the minimum cost incurred by any agent, i.e.,$\FR(T) \coloneqq \frac{\max_{v_i\in V}cost_T(i)}{\min_{v_i\in V}cost_T(i)}$.

Finally, we introduce some additional notation for arguing about the designed networks $T(\stp)$. (See Fig.~\ref{fig:notation} for an illustration). 
For our analysis we use directed paths in $T(\stp)$ which start at some non-root node $z$ and end at the root $r$. Let $P$ be such a path of length $\ell \in \mathbb{N}$. We denote by $u_j^P$ the node on $P$ which is at distance $j$ to $r$, hence the root $r$ is denoted by $u_0^P$ and node $z$ by $u_\ell^P$. Moreover, let $T_j^P \coloneqq T(u_j^P)$ and we use $\dd_j^P$ for the in-degree of a node with distance $j$ from the root $r$ on path $P$, hence, $\dd_{j}^P \coloneqq \indeg_{T(\stp)}(u_j^P)$.  
We omit the reference to path $P$ whenever it is clear from the context.

 \subsection{Related Work}
 Our model is closely related to several models that have been intensively studied. 
 
 We start with the (fair) single-source connection game~\cite{ADTW08,ADKTWR} which we already briefly discussed in the introduction. 
 The key feature of this game is that agents strategically select a set of edges to connect their respective terminals. The cost of each edge is shared among all the agents who selected the respected edge. While in~\cite{ADTW08} and later also in~\cite{H09} arbitrary cost sharing is considered, the paper~\cite{ADKTWR} focuses on fair cost sharing which can be derived from the Shapley value~\cite{MS01}. 
 For this Anshelevich et al.~\cite{ADKTWR} show that stable networks always exist since the game is a potential game~\cite{MS96}, additional they prove that the PoA is $n$ and the PoS is upper bounded by $H_n$, where~$H_n$ is the $n$-th harmonic number. For a given directed host graph this bound on the PoS is tight but the case for undirected host networks is still a major open problem. 
 More is known for single-source connection games on undirected networks. Chekuri et al.~\cite{CCLNO07} show that the PoA is in $\mathcal{O}(\sqrt{n}\log^2n)$ if the agents join the game sequentially and play their respective best response. A PoS in $\mathcal{O}(\log\log n)$ was proven by Fiat et al.~\cite{fiat2006price} for the special case where all nodes of the given network correspond to a terminal of some agent. Finally, Bilò et al.~\cite{bilo2014price} prove a constant PoS for the fair single-source connection game on undirected networks. Moreover, Albers and Lenzner~\cite{AL13} show that the optimum is a $H_n$-approximate Nash equilibrium for the fair single-source connection game. In contrast to our model, the cost of an edge in the (fair single-source) connection game is given via a positively weighted host network. Hoefer and Krysta~\cite{HK05} investigate a variant with edge weights derived from a geometry.      
 
 Also selfish spanning tree games~\cite{GM08} are close to our model and we already briefly discussed them in the introduction. The key difference to our model is that a weighted complete network is given and that the cost of an agent is defined differently. Gourvès and Monnot~\cite{GM08} define three variants of the agents' cost function: either it is the weight of the first edge on the path to the common root $r$, or the minimum or maximum weight edge on the entire path towards $r$. Cost sharing is not considered. The authors prove bounds on the PoA which vary from unbounded to 1 depending on the exact setting.
 The games in~\cite{GM08} are inspired by the classical problem of allocating the cost of a spanning tree among its nodes by Claus and Kleitman~\cite{CK73} and its variant from cooperative game theory considered by Bird~\cite{Bird76}. Later, Granot and Huberman~\cite{GH81,GH84} considered minimum cost spanning tree games and different cost allocation protocols for this have been considered by Escoffier et al.~\cite{EGMM12}. The key difference of all these models to our model is that a cooperative game is considered which is a stark contrast to our non-cooperative setting. Also game-theoretic topology control problems are related to spanning tree games and our model. Eidenbenz et al.~\cite{EKZ06} consider a setting where a set of agents which correspond to wireless devices want to connect terminal nodes, whereas Mittal et al.~\cite{MBS08} consider wireless access point selection by selfish agents. 
 
 Also classical network formation games~\cite{JW96,BG00,Fab03} are related to our model. There the agents correspond to nodes in a network and every agent buys a set of incident edges to connect to other agents. The goal of each agent is to create a connected network and to occupy a central position in this network. For the influential network creation game of Fabrikant et al.~\cite{Fab03}, that has a parameter $\alpha$ for the trade-off between edge cost and distance costs, the PoA was shown to be constant for almost all values of $\alpha$~\cite{De07,AM19}. For high values of $\alpha$ all equilibrium networks of these games are known to be trees~\cite{MS10,MMM13,BL20}. A variant of the network creation game where agents can only buy a single edge was considered by Ehsani et al.~\cite{Ehs15}. Most notably, the topology dependent edge costs that we employ in our model were proposed by Chauhan et al.~\cite{CLMM17} for the network creation game~\cite{Fab03}. To the best of our knowledge, this is the only setting where topology dependent edge costs have been considered.

\subsection{Our Contribution}
We study a novel game-theoretic model for the formation of a tree network which is related to the well-known fair single-source connection game by Anshelevich et al.~\cite{ADKTWR, ADTW08} and to selfish spanning tree games by Gourvès and Monnot~\cite{GM08}. The key difference of our model is that we consider dynamic edge costs which depend on the topology of the created spanning tree. In particular, the cost of an edge is equal to the in-degree of its endpoint. This specific choice was proposed in \cite{CLMM17} for the classical network creation game~\cite{Fab03} and we transfer this idea to the Network Design domain. Our analysis holds for any edge cost function of the form $\alpha$ times the in-degree of the target node, for any constant $\alpha$. However, our general approach is valid also for edge cost functions that depend non-linearly on the degrees of the involved nodes.   

Regarding the existence of stable trees we show that our model is in stark contrast to the models in~\cite{ADKTWR,Fab03,GM08} since in our model stable trees may not exist. In particular, we show that our game has no NE for $n=16$ and $n=18$ which implies that the \SNCG cannot admit a potential function. This is contrasted with a proof that for infinitely many $n$ stable trees do exist, and we conjecture that we have found all examples for NE non-existence. Towards investigating the quality of the equilibrium networks of our model, we first provide a rigorous study of the structural properties of stable trees. We show that every stable tree consists of stable subtrees and that the height of any stable tree is in $\mathcal{O}\left(\frac{\log n}{\log\log n}\right)$. For the root $r$, which turns out to be the node with the highest in-degree in any stable network, we show that its in-degree is between $\Omega\left(\frac{\log n}{\log\log n}\right)$ and $2^{\mathcal{O}\big(\sqrt{\log n}\big)}$. This shows that the maximum internal coordination overhead of a single node in any stable tree is rather small.

Our main results are on the quality of equilibrium trees. By using the established structural properties and a connection to the Riemann zeta function we obtain an upper bound on the PoA of~$8.62$ which is contrasted with a lower bound of $2.4317$. For the PoS we derive a lower bound of~$\frac{7}{5}-\varepsilon$. Moreover, we give for an infinite number of values for $n$ an upper bound of $2.83$ on the PoS. Regarding the Fairness Ratio, we first show that the socially optimal tree is rather unfair, i.e., having a Fairness Ratio of $n\cdot H_n$. In contrast, we prove that any equilibrium tree has a Fairness Ratio in $o(n)$. 

This shows that stable trees have only slightly higher social cost compared to the social optimum. In particular, on average every agent pays only a constant factor more than the trivial lower bound for any spanning tree. At the same time stable trees are more fair, have low height and low in-degrees.  

We conclude with a brief discussion of the path version extension of our model, where agents select paths as strategies as in~\cite{ADTW08,ADKTWR}. This extension seems promising for future work since we show that allowing a richer strategy space yields a larger set of equilibria and we give equilibria for $n=16$ and $n=18$. Hence, in the path-version equilibria may always exist, but the PoA could be higher.   

\section{Structure and Properties of Equilibrium Trees}\label{sec:structure}
It is clear that each agent can compute her best response in polynomial time as the number of possible strategies for an agent is $n$, and the agent can easily compute her cost in linear time. In the following we show that any stable tree consists of stable subtrees, we prove an upper bound of $\mathcal{O}\left(\frac{\log n}{\log \log n}\right)$ to the number of edges of any leaf-to-root path of any stable network, and in the end, we provide bounds on the degree of the root. We start with the statement that any stable tree consists of stable subtrees.

\begin{lemma}\label{thm:subtreeNE}
	If $T$ is stable, then any subtree $T(x)$ is stable in the corresponding subgame. 
\end{lemma}

\begin{proof}
	Consider $T$. Assume to the contrary that there is a subtree $T(x)$ which is not stable. Then there is an agent $y\in T(x)$ which can improve her strategy by swapping her edge $(u,u_1)$ with an edge $(u,v)$. Let $u=u_0,u_1,\ldots,u_k=x$ be the path from $u$ to $x$ and $v=v_0,\ldots,v_m=x$ be the path from~$v$ to $x$. Let $T'(x)$ be the subtree obtained after $u$ changed her strategy towards $(u,v)$. Then the new strategy implies the following difference of the costs equals to
	\begin{align*}
	&\ \cost_{T'(x)}(u) - \cost_{T(x)}(u)  \\
	 = &\ \frac{\indeg(v_0)+1}{|T(u)|}+\frac{\indeg(v_1)}{|T(u)|+|T(v)|}+\ldots +\frac{\indeg(x)}{|T(v_{m-1})|+|T(u)|}- \frac{\indeg(u_1)}{|T(u)|}-\ldots -\frac{\indeg(x)}{|T(u_{k-1})|} .
	\end{align*}
	Since $T$ is stable, agent $u$ cannot improve her strategy by the same swap. Let $T'$ be a tree obtained from $T$ after $u$ changed her strategy towards $(u,v)$, and let $x=x_0, x_1,\ldots,x_l=r$ be the path from $x$ to the root $r$. Then we have
	\begin{align*}
	0 &\leq \cost_{T'}(u) - cost_{T}(u)\\
	&=\frac{\indeg(v_0)+1}{|T(u)|}+\frac{\indeg(v_1)}{|T(u)|+|T(v)|}+\ldots +\frac{\indeg(x)}{|T(v_{m-1})|+|T(u)|} +\frac{\indeg(x_1)}{|T(x)|}+\ldots+\frac{\indeg(r)}{|T(x_{l-1})|}\\
	&- \frac{\indeg(u_1)}{|T(u)|}-\ldots -\frac{\indeg(x)}{|T(u_{k-1})|}-\frac{\indeg(x_1)}{T(x)}-\ldots-\frac{\indeg(r)}{|T(x_{l-1})|}\\
	&=\cost_{T'(x)}(u) - cost_{T(x)}(u).
	\end{align*}
	Hence, an agent obtains the same cost improvement in a restricted game as in the original game because a strategy change does not affect the load and the in-degree of edges outside of the considering subtree. Therefore, since $T$ is stable, every $T(x)$ is stable as well in the corresponding subgame. 
\end{proof}

\noindent Next, we will consider the height of a stable network and need  the following technical lemmas.
\begin{lemma}\label{lemma:degree_formula}
	Let $k \in \mathbb{N}$ be the length of a fixed leaf-to-root-path $P$ in a stable network $T$. Then, for every $1 < i < k$, $\displaystyle{\dd_{i-1} \geq \frac{|T_i|}{|T_i|-|T_{i+1}|}(\dd_{i}-1)}$.
\end{lemma}

\begin{proof}
	As agent $u_{i+1}$ has no incentive to swap the edge $(u_{i+1},u_i)$ with the edge $(u_{i+1},u_{i-1})$, it follows that
	$
	\frac{\dd_i}{|T_{i+1}|} + \frac{\dd_{i-1}}{|T_i|} \leq \frac{\dd_{i-1}+1}{|T_{i+1}|},
$
	the claim follows.
\end{proof}

\noindent Since $|T_{i+1}| > 0$, Lemma~\ref{lemma:degree_formula} yields that the sequence $\dd_0$, $\dd_1,\dots, \dd_k$, is monotonically decreasing. 
\begin{corollary}\label{corollary:degree_seguence_monotonic}
	Let $k \in \mathbb{N}$ be the length of a leaf-to-root-path $P$ in a stable network $T$. Then, for every $1 < i < k$, $\dd_{i} \geq d_{i+1}$.
\end{corollary}
The next lemma shows that the in-degree of nodes strictly decreases after a constant number of hops.
\begin{lemma}
	Let $k \in \mathbb{N}$ be the length of a fixed leaf-to-root-path $P$ in a stable network $T$. Then, for every subtree $T(v)$ of $T$ with $|T(v)| > 4$ and for every $1 < i < k-2$ we have $d_{i-1} > d_{i+1}$.
	\label{thm:same_degree}
\end{lemma}

\begin{proof}
	By Corollary~\ref{corollary:degree_seguence_monotonic} we know that $\dd_{i} \geq d_{i+1}$. Hence, the in-degree sequence cannot decrease towards the root. Assume to the contrary that there is a subtree $T(v)$ of $T$ with $|T(v)| > 4$ and a subpath $P' = (u,x)$, $(x,y)$, $(y,v)$ with $P' \subseteq P$ and $\indeg(x) = \indeg(y) = \indeg(v) = d$. Assume that $T(u)$ is the largest subtree of $T(x)$ and $T(x)$ the largest subtree of $T(y)$. We will later show that there exists always such a node $x$ with $\indeg(x) = d$. Since $T$ is stable $u$ cannot improve by swapping her edge $(u,x)$ with the edge $(u,v)$. Therefore, 
	\begin{eqnarray*}
		0  & \geq &\left( \frac{d}{|T(u)|} +  \frac{d}{|T(x)|} + \frac{d}{|T(y)|} \right) - \frac{d+1}{|T(u)|}  \\
		& \geq  &   \frac{d}{|T(u)| \cdot d + 1}  + \frac{d}{|T(x)| \cdot d + 1} - \frac{1}{|T(u)|}  \\
		& \geq  &   \frac{d}{|T(u)| \cdot d + 1}  + \frac{d}{(|T(u)| \cdot d + 1) \cdot d + 1} - \frac{1}{|T(u)|} \\
		& = &  \frac{|T(u)| \cdot d}{|T(u)| \cdot \left(|T(u)| \cdot d + 1\right)}  + \frac{d}{|T(u)| \cdot d^2 + d + 1} - \frac{|T(u)| \cdot d + 1}{|T(u)| \cdot \left(|T(u)| \cdot d + 1\right) }  \\
		& = &  \frac{d}{|T(u)| \cdot d^2 + d + 1} - \frac{ 1}{|T(u)|^2 \cdot d + |T(u)| } \\
		& = &  \frac{d \cdot (|T(u)| - 1) \cdot (|T(u)| \cdot d + 1) -1}{|T(u)| \cdot (|T(u)| \cdot d + 1) \cdot (|T(u)| \cdot d^2 + d + 1)} .
	\end{eqnarray*}
	Since $|T(u)|$ and $d$ are non-negative integer values the denominator is positive, hence, we have to show that $d \cdot (|T(u)| - 1) \cdot (|T(u)| \cdot d + 1) -1 \leq 0$.
	However, the inequality does not hold if $|T(u)| \geq 2$ since every multiplier is larger than $1$, and therefore the product $d \cdot (|T(u)| - 1) \cdot (|T(u)| \cdot d + 1)$ is strictly larger than $1$. 
	
	If $|T(u)| = 1$ it implies $d > 1$ since $|T(v)| > 4$. Since $T(u)$ is the largest subtree of $T(x)$ there has to be a leaf node $u'$ with an edge $(u',x)$. However, the costs of agent $u$ in $T$ are equal to $d + \alpha$ for $\alpha \geq 0$, while the cost of $u$ swapping her edge $(u,x)$ with $(u,u')$ is equal to $1+ \frac{d-1}{2} + \alpha$, which is an improvement for $d > 1$.	
	Hence, $T$ cannot be stable.
	
	To show that there exists a node $x$ with $\indeg(x) = d$ as the root of the largest subtree of~$T(y)$, we assume to the contrary that there is another subtree $T(x')$ of $T(y)$ with $\indeg(x') < d$ and $|T(x')| > |T(x)|$. Since $T$ is stable $u$ cannot improve by swapping her edge $(u,x)$ with the edge $(u,x')$. However,  
	$
	0  \geq   \frac{d}{|T(u)|} +  \frac{d}{|T(x)|} -  \frac{\indeg(x')}{|T(u)|}-  \frac{d}{|T(x')|}
	$
	does not hold in this case. This completes the proof.
\end{proof}

\noindent In the following we investigate upper and lower bounds on the in-degree of the root in stable trees. More precisely, we show an upper bound of $2^{\mathcal{O}(\sqrt{\log n})}$ and a lower bound of $\Omega(\log n/\log\log n)$. 

\begin{theorem}\label{thm:LB_deg_of_the_root}
	The in-degree of the root in a stable network $T$ is at least $\frac{\ln\left(4\sqrt{n/5}\right)}{\ln\ln\left(4\sqrt{n/5}\right)}$.
\end{theorem}

\begin{proof}
	Consider a stable network $T$ of height $h$. Consider the case when $h\geq 4$, otherwise, $T$ can be a path, i.e., the minimal in-degree of the root is 1. For $h\geq 4$, we observe that the root of $T$ has a minimal in-degree if the sequence of the in-degrees $\dd^P_h, \ldots, \dd^P_0$ of the nodes in the longest leaf-to-root path $P=v^P_h,\ldots, v^P_0$ is minimally increasing on the way to the root, i.e., $\dd^P_i\leq \dd^P_{i-1} +a$ where $a\geq 0$ is the smallest possible value. Hence, by Lemma~\ref{thm:same_degree}, $\dd^P_0$ is minimal if $\dd^P_h=0, \dd^P_{h-1}=\dd^P_{h-2}=\dd^P_{h-3} =1$ and, for $0\leq i\leq h-4$, $\dd^P_i= \dd^P_{i-1}+1$ if $h-i$ is even, and $\dd^P_i= \dd^P_{i-1}$ if $h-i$ is odd.  Then $\dd^P_0\geq \lfloor\frac{h-2}{2}\rfloor+1$, and we need to get a lower bound for $h$ with respect to $n=|T|$ to finish the proof.
	
	Note that the size of the tree $T$ is maximal if every subtree $T_{i+1}$ rooted at a child of  a node $v^P_i$ in the $v^P_h$-$v^P_0$ path has maximal size. Hence, the in-degree of the root of $T_{i+1}$ is maximal, i.e., it equals $\max\{\dd^P_i-1, \dd^P_{i-1}-1\}$. This implies that $T$ is a balanced tree $BT=(0,1,1,1,2,2,3,3,\ldots,\dd_0)$, i.e., each leaf-to-root path in $T$ corresponds to the in-degree sequence  $\dd_h:=\dd^P_h,\ldots, \dd_0^P:=\dd_0$. Without loss of generality assume $h-4$ is even. Then we have:
	\begin{align*}
		n =|T|\leq |BT| = \sum\limits_{i=0}^{h-1}\prod\limits_{j=0}^i \dd_j
		&= \dd_0 + \dd_0^2 + \dd_0^2(\dd_0-1)+\ldots + \dd_0^2(\dd_0-1)^2\cdot\ldots\cdot\left(\dd_0 -\frac{h-4}{2}\right)^2\cdot 4\\
		& < \dd_0 + \dd_0^2 + \dd_0^3+\ldots + \dd_0^{h-5}+4\cdot \dd_0^{h-4}= \frac{\dd_0(\dd_0^{h-4}-1)}{\dd_0-1}+3\cdot \dd_0^{h-4}\\
		&< \dd_0^{h-4}\left(\frac{\dd_0}{\dd_0-1}+3\right) \leq 5\dd_0^{h-4} \leq \frac{5}{16}\dd_0^{2\dd_0}.
	\end{align*}
	Therefore, $\dd_0 > \frac{\ln(16n/5)}{2W\left(\frac{\ln(16n/5)}{2}\right)}>\frac{\ln\left(4\sqrt{n/5}\right)}{\ln\ln\left(4\sqrt{n/5}\right)}$, where $W(x)$ is the Lambert function.
\end{proof}

\noindent To give an upper bound on the in-degree of the root, we first have to provide the following technical lemmas. The first technical lemma bounds the in-degree of the parent of any leaf.

\begin{lemma}\label{lem:deg_parent_of_leaf}
	In a stable network $T$ the in-degree of the parent of any leaf is 1.
\end{lemma}

\begin{proof}
	Consider a node $v$ in $T$ which has two children $v_1$ and $v_2$ such that $v_1$ is a leaf node. 
	Then~$v_2$ can swap to the leaf $v_1$ and improve its cost by at least \[\frac{\indeg(v)}{|T(v_2)|} - \left(\frac{1}{|T(v_2)|}+\frac{\indeg(v)-1}{|T(v_2)|+1}\right)=(\indeg(v)-1)\left(\frac{1}{|T(v_2)|}-\frac{1}{|T(v_2)|+1}\right)>0.\qedhere\]
\end{proof}

\noindent The second technical lemma shows how the in-degrees of two sibling nodes are related.

\begin{lemma}\label{thm:deg_up} 
	Consider a subtree $T(x)$ of a stable network $T$. 
	Then $\indeg(x)\leq \indeg(v)\cdot\left(1+\frac{|T(u)|}{|T(v)|}\right)+1$, where $v$ and $u$ are different children of $x$. 
\end{lemma}

\begin{proof}
	Consider two children $u$ and $v$ of the root $x$ in the subtree $T(x)$. %W.l.o.g. assume $|T_u|\leq |T_v|$. 
	Since $T$ is stable, $u$ cannot improve her strategy by swapping the edge $(u,x)$ with the edge $(u,v)$. Let $T'(x)$ be the subtree obtained after $u$ changed her strategy towards $(u,v)$. This implies that 
	\begin{align*}
	0 &\leq \cost_{T'(x)}(u) - cost_{T(x)}(u)=\frac{\indeg(v)+1}{|T(u)|}+\frac{\indeg(x)-1}{|T(u)|+|T(v)|}-\frac{\indeg(x)}{|T(u)|}\\
	&= \frac{1}{|T(u)|\left(|T(u)|+|T(v)|\right)}\left(\left(1+\frac{|T(u)|}{|T(v)|}\right)\cdot \indeg(v)+1-\indeg(x)\right)
	\end{align*}
	Therefore,  $\indeg(x)\leq \indeg(v)\cdot\left(1+\frac{|T_u|}{|T_v|}\right)+1$.
\end{proof}

\noindent From Lemma~\ref{thm:deg_up}, we derive the following remark and corollary.

\begin{remark}
	Consider a subtree $T(x)$ in a stable network $T$. Then $\indeg(x)\leq 2\cdot \indeg(v)+1$, where~$v$ is a root of the second smallest subtree of $T(x)$. 
\end{remark}

\begin{corollary}\label{cor:min_deg_of_child}
	If $T$ is a stable network, then every node $u$ in $T$ has at least $\indeg(u)-1$ children of in-degree at least $(\indeg(u)-1)/2$.  
\end{corollary}

\noindent Now we can prove an upper  bound to the in-degree of the root of any stable tree.

\begin{theorem}\label{thm:UB_deg_of_the_root}
	The in-degree of the root in a stable network $T$ is $2^{\mathcal{O}\left(\sqrt{\log{n}}\right)}$.
\end{theorem}
\begin{proof}
	Consider a stable tree $T$ of height $h$. Let $v_h,\ldots,v_0$ be a path from a leaf to the root. Note that the in-degree of the root $v_0$ is maximal if the in-degree of each node in the $v_h$-$v_0$-path is maximal, i.e., by Lemma~\ref{thm:deg_up} and~\ref{lem:deg_parent_of_leaf}, it corresponds to the in-degree sequence $D :=(0,1,\dd_{h-2}, \ldots,\dd_0)$, where $\dd_{i-1}=2\dd_i+1$. %By Lemma, $\dd_{h-1}=1$.
	
	Next, we show that nodes at distance $h-2$ from the root can have an in-degree of at most $2$. Assume to the contrary that there is a node $u$ having an edge to a node $x$ such that $\indeg(x)=3$ and $x$ is at distance $h-2$ from the root $v_0$. As we have proved above, the in-degree of all children of $x$ is at most~$1$. Thus, $u$ can swap to any leaf node of the subtree $T(x)$. Let $T'$ be the tree obtained after $u$ swapped. If $u$ swaps to a child of $x$, it decreases its cost by $cost_{T}(u) - cost_{T'}(u) = \frac{3}{2}-\frac{1}{2}-\frac{2}{3} > 0$, i.e., it is an improving move. The swap to a leaf node at distance $2$ from $x$ implies an improvement by $cost_{T}(u) - cost_{T'}(u) = \frac{3}{2}-\frac{1}{2}-\frac{1}{3} - \frac{2}{4}> 0$, i.e., it is an improvement. Since $T$ is stable, we get a contradiction. Thus, $D=(0,1,2,5,11,\ldots,\dd_0)$, i.e., \begin{equation}\dd_i=3\cdot 2^{h-i-2}-1 \text{ for } i\leq h-3, \text{ where } \dd_h=0, \dd_{h-1}=1, \dd_{h-2}=2.\label{eq:recurr_deg}\end{equation}  
	We now estimate the minimum possible number of nodes in the tree $T$. 
	By Corollary~\ref{cor:min_deg_of_child} if the in-degree of a node is equal to $k$, then it has at least $k-1$ children with an in-degree of at least $(k-1)/2$. Thus, starting from the root, the in-degrees of the nodes on each level decrease no more than twice. 
	Hence, the total size of the tree is  at least \[\sum\limits_{i=1}^{h-1}\left(\dd_i\prod\limits_{j=0}^{i-1} (\dd_j-1)\right) > \sum\limits_{i=1}^{h-1}\prod\limits_{j=0}^{i-1} 2^{h-j-2} > 2^{\sum_{j=0}^{h-3}(h-j-2)}=2^{\frac{(h-1)(h-2)}{2}},\] where $h$ is the height of $T$. Thus, $h<\frac{3+\sqrt{1+8\log n}}{2}$. With equation~(\ref{eq:recurr_deg}), this implies $\dd_0\in 2^{O\left(\sqrt{\log n}\right)}.$ \qedhere
	
\end{proof}
Now we are able to show that the length of any node-to-root path is $\mathcal{O}\left(\frac{\log n}{\log \log n}\right)$.
\begin{theorem}
	If $T$ is a stable network, then its height $h\in \mathcal{O}\left(\frac{\log n}{\log \log n}\right)$.
\end{theorem}
\begin{proof}
	Consider a leaf-to-root path $P$ in $T$. We show that there are $\mathcal{O}\left(\frac{\log n}{\log \log n}\right)$ indices such that $|T^P_i|-|T^P_{i+1}| \geq \sqrt[3]{\log n} \cdot |T^P_{i+1}|$ and $\mathcal{O}\left(\frac{\log n}{\log \log n}\right)$ indices such that $|T^P_i|-|T^P_{i+1}| < \sqrt[3]{\log n} \cdot |T^P_{i+1}|$.
	
	Let $k$ be the number of indices $i$ that satisfy $|T^P_{i}| - |T^P_{i+1}| \geq \sqrt[3]{\log n} \cdot |T^P_{i+1}|$. 
	Then we have that $|T^P_{i}| = |T^P_{i}|-|T^P_{i+1}|+|T^P_{i+1}| >  \sqrt[3]{\log n} \cdot |T^P_{i+1}|$.
	Since $|T^P_i| > |T^P_{i+1}|$ for every $i$, and because $|T^P_i| \leq n$, we have that $|T^P_0| > (\log n)^{k/3}$ and $|T^P_0| = n+1$, from which we derive $(\log n)^{k/3} \leq n$, i.e., $k = \mathcal{O}\left(\frac{\log n}{\log \log n}\right)$.
	
	By Lemma~\ref{thm:same_degree} and Corollary~\ref{corollary:degree_seguence_monotonic}, there are $\mathcal{O}(\sqrt[3]{\log n})=\mathcal{O}\left(\frac{\log n}{\log \log n}\right)$ indices $i$ such that $\dd^P_i \leq 4 \sqrt[3]{\log n}$.
	We now prove that there are $\mathcal{O}\left(\frac{\log n}{\log \log n}\right)$ indices such that $|T^P_{i}| - |T^P_{i+1}| < \sqrt[3]{\log n} \cdot |T^P_{i+1}|$
	and such that $\dd^P_i \geq 4 \sqrt[3]{\log n}$. 
	By Lemma~\ref{lemma:degree_formula} and using the fact that $\dd^P_i \geq 4 \sqrt[3]{\log n}$ and $n\geq 2$, we have that
	\[
	\dd^P_{i-1}\geq \frac{|T^P_i|}{|T^P_i|-|T^P_{i+1}|}(\dd^P_{i}-1) \geq \frac{1+\sqrt[3]{\log n}}{\sqrt[3]{\log n}}(\dd^P_i-1) \geq \sqrt{\frac{1+\sqrt[3]{\log n}}{\sqrt[3]{\log n}}}\dd^P_i.
	\]
	By Corollary~\ref{corollary:degree_seguence_monotonic} we have that $\dd^P_{i-1} \geq \dd^P_i$ for every $i$. Hence $d_0 \geq \left(\frac{1+\sqrt[3]{\log n}}{\sqrt[3]{\log n}}\right)^{k/2}$. 
	Moreover, by Theorem~\ref{thm:UB_deg_of_the_root}, $\dd^P_0 \leq 2^{\alpha \sqrt{\log n}}$ for some constant $\alpha > 0$. 
	As a consequence, we have that 
	$\left(\frac{1+\sqrt[3]{\log n}}{\sqrt[3]{\log n}}\right)^{k/2} \leq 2^{\alpha \sqrt{\log n}}$, i.e.,  $2^{\frac{k}{2}\log \frac{1+\sqrt[3]{\log n}}{\sqrt[3]{\log n}}} \leq 2^{\alpha \sqrt{\log n}}$, which implies,
	$k \leq 2\alpha \frac{\sqrt{\log n}}{\log \left(1+1/\sqrt[3]{\log n}\right)}.
	$
	
	We complete the proof by showing that 
	$
	\frac{\sqrt{\log n}}{\log \frac{1+\sqrt[3]{\log n}}{\sqrt[3]{\log n}}} \leq 3 \frac{\log n}{\log \log n}$, for large enough $n$,
	i.e., we have to show that 
	$\log\log n \leq 3 \sqrt{\log n} \cdot \log \frac{1+\sqrt[3]{\log n}}{\sqrt[3]{\log n}}.
	$
	Let $M=\sqrt[3]{\log n}$. We have to prove that
	\begin{equation}\label{eq:log_formula}
	\log M \leq M^{3/2} \log \frac{1+M}{M}=\log \left( \frac{1+M}{M}\right)^{M^{3/2}}.
	\end{equation}
	By Bernoulli's inequality, $\left(1+\frac{1}{M}\right)^{M^{3/2}}\geq 2^{M^{1/2}}\geq M$ for $M\geq 16$. Thus, inequality (\ref{eq:log_formula}) is satisfied.\qedhere
	
\end{proof}

\section{Existence of Equilibrium Trees}\label{sec:existence}
In this section we analyze whether the \SNCG admits equilibrium trees for all agent numbers $n$. We first show that in general equilibrium existence is not guaranteed since for $n=16$ and $n=18$ no stable tree exists. We contrast this negative result with a NE existence proof for infinitely many agent numbers $n$. This positive result is achieved for so-called balanced trees, i.e., trees where all nodes with the same distance to the root have the same in-degree. We believe that our positive results can be strengthened to proving that stable trees exist for all $n$ except $n=16$ and $n=18$, and we leave this as an intriguing open problem. 
Figure~\ref{fig:sample_trees} shows sample equilibrium trees for small $n$. 
\begin{figure}[h]
 \centering
 \includegraphics[width=0.9\textwidth]{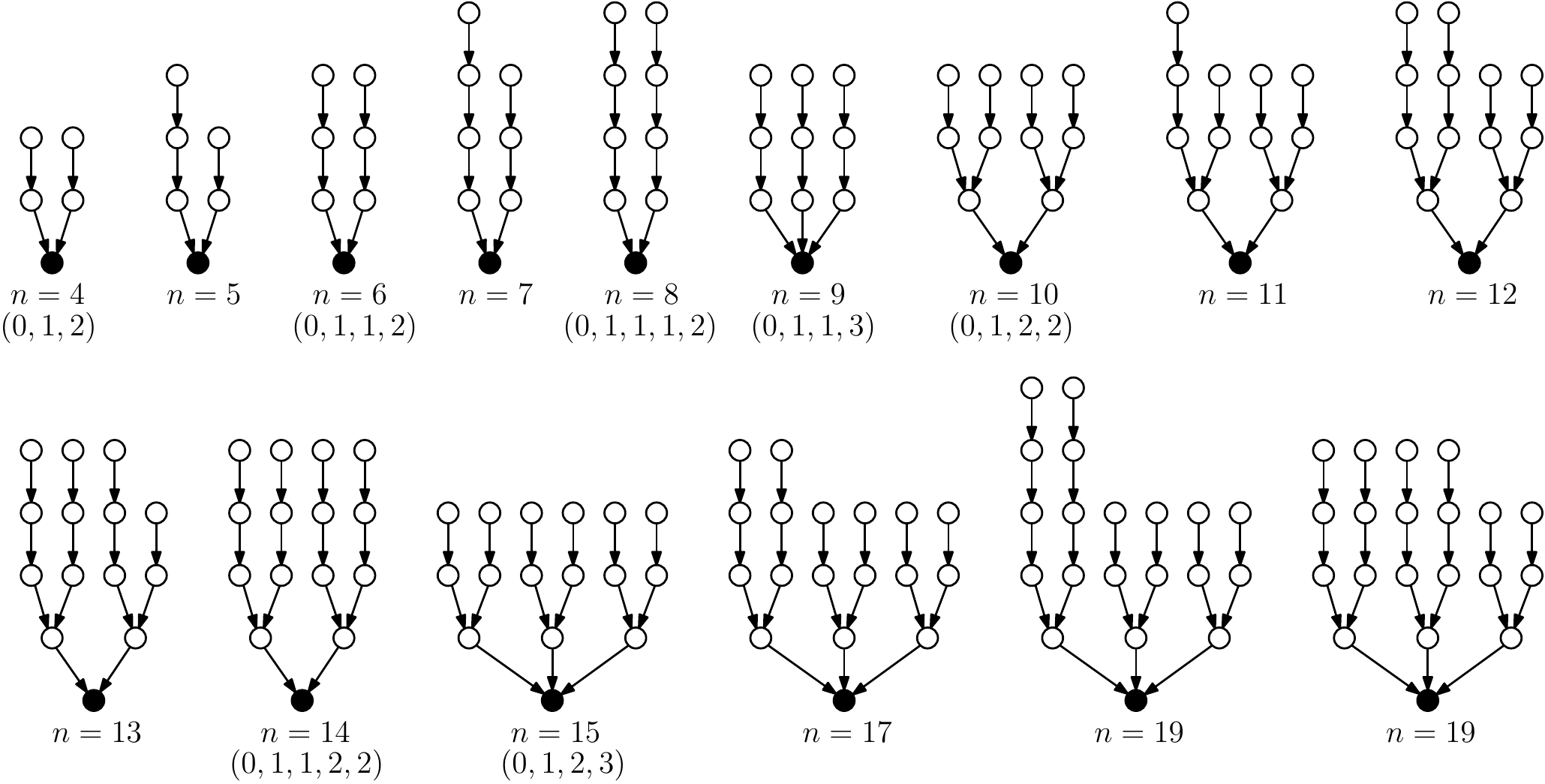}
 \caption{Sample equilibrium trees for $n=4$ to $n=19$. All depicted trees for $n<19$ are the unique equilibria for the respective $n$. For $n=19$ two equilibrium trees exist. No stable tree exists for $n=16$ and $n=18$.
 The stable trees for $n=4,6,8,9,10,14,15$ are balanced trees and are annotated with their identifying in-degree sequence of all leaf-to-root paths. (See Section~\ref{sec:balanced_trees} for definitions.) }
 \label{fig:sample_trees}
\end{figure}

\begin{theorem}\label{thm:non_existence_stable_tree_for_n_equal_16}
	For $n = 16$ there exists no stable network. 
\end{theorem}

\begin{proof}
	Assume for the contrary that there is a stable network $T$. We first show that the in-degree of the root $r$ in $T$ is equal to $2$ or $3$. By Lemma \ref{thm:same_degree} we already know that $\indeg(r) \geq 2$ in~$T$.
	
	Assume that $\indeg(r) \geq 4$ and that there is a subtree $T(x)$ with an edge $(x,r)$ and $|T(x)| \leq 3$. Hence, apart from $x$, there can be at most two additional nodes, $x'$ and $x''$, in $T(x)$. By Lemma \ref{thm:subtreeNE} we know that $T(x)$ is a path and therefore $\indeg(x) < 2$. Otherwise $T(x)$ is not stable since $x'$ can swap the edge $(x',x)$ with $(x',x'')$ and decrease her cost from $2$ to $\frac32$. Hence, $\indeg(x) = 0$ or $\indeg(x) = 1$.
	
	Consider an agent $y \neq x$ with an edge $(y,r)$ and a strategy change from $(y,r)$ to $(y,x)$. The current costs of $y$ are $\frac{\indeg(r)}{|T(y)|}$. If $\indeg(x) = 0$ the swap results in cost of $\frac{1}{|T(y)|} + \frac{\indeg(r)-1}{|T(y)|+1}$. This is an improvement for $y$. If $\indeg(x) = 1$ the costs for $y$ would be at most $\frac{2}{|T(y)|} + \frac{\indeg(r)-1}{|T(y)|+2}$ which is an improvement if $|T(y)| \leq 4$. However, since $\indeg(r) \geq 4$ and $|T(x)| \geq 2$ there has to be an agent $y$ with an edge $(y,r)$ and $|T(y)| \leq 4$. Therefore, for all agents $x$ with $(x,r)$ it holds that $|T(x)| \geq 4$ which is only possible if $\indeg(r) = 4$ and $|T(x)| = 4$. Let $u$, $v$, $w$ and $x$ be the nodes which belong to $T(x)$.
	
	There are three different remaining cases how the subtrees $T(x)$ can look like: \begin{itemize}
		\item $\indeg(x) = 3$: There are two leaf nodes $u$ and  $v$ with an edge to $x$. Swapping $(u,x)$ to $(u,v)$ decreases the cost of $u$ from $4$ to $3$.
		\item $\indeg(x) = 2$: $T(x)$ includes the edges $(u,x)$, $(v,x)$ and $(w,u)$. Agent $v$ can reduce her cost from $3$ to $\frac{17}{6}$ by swapping $(v,x)$ to $(v,w)$.
		\item $\indeg(x) = 1$: Because of Corollary \ref{corollary:degree_seguence_monotonic},  $T$ includes two paths $(u_1, v_1)$, $(v_1, w_1)$, $(w_1, x_1)$, $(x_1, r)$ and $(u_2, v_2)$, $(v_2, w_2)$, $(w_2, x_2)$, $(x_2, r)$. Agent $x_1$ can improve by swapping her edge to $(x_1, x_2)$ which reduces her costs from $1$ to $\frac78$.
	\end{itemize}
	Hence, $T$ with $\indeg(r) \geq 4$ cannot be stable.
	
	We first consider the case $\indeg(r) = 2$.
	By Corollary \ref{corollary:degree_seguence_monotonic} and Lemma \ref{thm:same_degree} we know that besides the two nodes $x'$ and $x''$ directly connected to the root with $(x',r)$ and $(x'',r)$, all agents~$i$ have an in-degree $\indeg(i) \leq 1$. Since Lemma \ref{thm:same_degree} bounds the maximum length of a simple path where all edges costs are equal $1$ by $4$, it holds that $\indeg(x') = \indeg(x'') = 2$ and $|T(x')| \leq 9$ and  $|T(x'')| \leq 9$, respectively. Otherwise it would be impossible to place all agents in the equilibrium tree~$T$. For the same reason there is at least one path of length $4$, $(t,u)$, $(u,v)$, ($v,w)$, with $(w,x')$ or $(w,x'')$ with $t$ as a leaf node with $\indeg(t) = 0$. The current costs of agent $v$ are at least $\frac13 + \frac24 + \frac29 = \frac {19}{18}$. However, agent~$v$ can improve by swapping her edge towards the root to $(v,r)$ and gains costs equal $1$.
	
	Therefore $r$ has an in-degree $\indeg(r) = 3$. For every subtree $T(x)$ with a root $x$ and an edge $(x,r)$ it holds that $|T(x)| \geq 2$. Otherwise another agent $x'$ with $(x',r)$ can improve by swapping her edge towards $x$ and reducing her costs from $\frac{3}{|T(x')|}$ to $\frac{1}{|T(x')|} + \frac{2}{|T(x')| + 1}$.
	
	Assume that for every subtree $T(x)$ with $x$ having the edge $(x,r)$ it holds that $|T(x)| \geq 5$, which implies that there are two subtrees of size $5$ and one subtree of size $6$. Because of Theorem \ref{thm:subtreeNE} there is only one possible tree $T$, see Figure~\ref*{img:n=16}. However, this is not stable since agent~$9$ can improve by swapping her edge towards $(9,6)$ and reduces her costs from $1+1+\frac35 = \frac{13}{5}$ to $1+\frac12+\frac23+\frac37 = \frac{109}{42}$.
	
	\begin{figure}
		\centering
		\begin{subfigure}[c]{0.3\textwidth}		
			\includegraphics[width=0.7\textwidth]{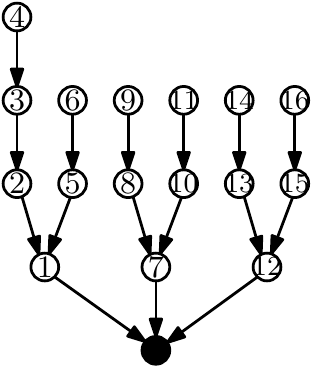}
			\subcaption{}
			\label{img:n=16}
		\end{subfigure}
		%\centering
		\hspace*{2em}
		\begin{subfigure}[c]{0.3\textwidth}
			\includegraphics[width=0.2\textwidth]{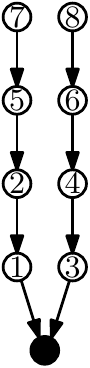}
			\subcaption{}
			\label{img:stablestate}
		\end{subfigure}
		%\centering
		\begin{subfigure}[c]{0.3\textwidth}
			\vspace*{1.5em}
			\includegraphics[width=0.35\textwidth]{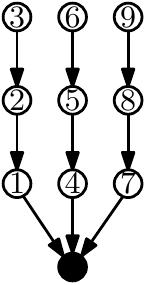}
			\subcaption{}
			\label{img:stablestate_10}
		\end{subfigure}
		\caption{(a) The only possible tree $T$ due to Theorem~\ref{thm:subtreeNE} where every subtree $T(x)$ with $x$ having an edge $(x,r)$ has size $|T(x)| \geq 5$. (b) The stable subtrees $T(x)$, $T(y)$ and $T(z)$ for $5 \leq T(i) \leq 9$ for $i \in \{x,y,z\}$. (c) The stable subtree $T(z)$ for $|T(z)| = 10$.  }
	\end{figure}

	Hence, there are three remaining cases. Let $T(x)$ be the smallest subtree with $x$ having the edge~$(x,r)$, i.e., with a direct edge to the root, $T(y)$ be the second smallest subtree with $(y,r)$ and~$T(z)$ be the largest subtree with $(z,r)$. It holds that $|T(x)| \leq 4$ and $T(x)$ is therefore a path.
	\begin{itemize}
		\item $|T(x)| = 2$: It holds that $|T(y)| > 3$ since otherwise agent $y$ can improve by swapping her edge towards the leaf node of $T(x)$ and having costs of $\frac{1}{|T(y)|} + \frac{1}{|T(y)|+1} +  \frac{2}{|T(y)|+2}$ instead of $\frac{3}{|T(y)|}$. Considering now the possible swap from $x$ towards $y$ leads to $\indeg(y) > 1$ and hence, $|T(y)| > 4$. Remember that due to Lemma~\ref{thm:subtreeNE} we know how the subtrees $T(x)$, $T(y)$ and $T(z)$ look like. We have three remaining cases:
		\begin{itemize}
			\item $|T(y)| = 5$ and $|T(z)| = 9$: This is not a stable tree since agent $8$, cf. Figure~\ref*{img:stablestate}, from $T(z)$ can improve by swapping her edge towards the leaf node of $T(x)$ and reduces her cost from $\frac83$ to $\frac52$.
			\item $|T(y)| = 6$ and $|T(z)| = 8$: This is not a stable tree since agent $7$, cf. Figure~\ref*{img:stablestate}, from $T(z)$ can improve by swapping her edge towards the leaf node of $T(x)$ and reduces her cost from $\frac{65}{24}$ to~$\frac52$.
			\item $|T(y)| = 7$ and $|T(z)| = 7$: This is not a stable tree since agent $6$, cf. Figure~\ref*{img:stablestate}, from $T(z)$ can improve by swapping her edge towards the leaf node of $T(x)$ and reduces her cost from $\frac{109}{42}$ to~$\frac52$.
		\end{itemize}
		\item $|T(x)| = 3$: It holds that $|T(y)| \neq 4$ since otherwise agent $x$ can reduce her current costs of $1$ to $\frac{20}{21}$ by swapping her edge towards $y$ and choose $(x,y)$. We have three remaining cases:
		\begin{itemize}
			\item $|T(y)| = 3$ and $|T(z)| = 10$: This is not a stable tree since agent $9$, cf. Figure~\ref*{img:stablestate_10}, from $T(z)$ can improve by swapping her edge towards the leaf node of $T(x)$ and reduces her cost from $\frac{14}{5}$ to~$\frac{31}{12}$.
			\item $|T(y)| = 5$ and $|T(z)| = 8$: This is not a stable tree since agent $7$, cf. Figure~\ref*{img:stablestate}, from $T(z)$ can improve by swapping her edge towards the leaf node of $T(x)$ and reduces her cost from $\frac{65}{24}$ to~$\frac{31}{12}$.
			\item $|T(y)| = 6$ and $|T(z)| = 7$: This is not a stable tree since agent $6$, cf. Figure~\ref*{img:stablestate}, from $T(z)$ can improve by swapping her edge towards the leaf node of $T(x)$ and reduces her cost from $\frac{109}{42}$ to~$\frac{31}{12}$.
		\end{itemize}
		\item $|T(x)| = 4$: We have three remaining cases:
		\begin{itemize}
			\item $|T(y)| = 4$ and $|T(z)| = 8$: This is not a stable tree since agent $7$, cf. Figure~\ref*{img:stablestate}, from $T(z)$ can improve by swapping her edge towards agent $x$ and reduces her cost from $\frac{65}{24}$ to $\frac{13}{5}$.
			\item $|T(y)| = 5$ and $|T(z)| = 7$: This is not a stable tree since agent $6$, cf. Figure~\ref*{img:stablestate}, from $T(z)$ can improve by swapping her edge towards agent $x$ and reduces her cost from $\frac{41}{14}$ to $\frac{13}{5}$.
			\item $|T(y)| = 6$ and $|T(z)| = 6$: This is not a stable tree since agent $5$ from $T(z)$ can improve by swapping her edge towards agent $4$ from $T(y)$ and reduces her cost from $\frac83$ to $\frac{109}{42}$, cf. Figure~\ref*{img:stablestate}.
		\end{itemize}
	\end{itemize}
	Observe that there are unique stable states for $n \in \{4, \ldots, 9\}$ and, together with Theorem~\ref{thm:subtreeNE}, this shows that there exists no stable tree $T$ for $n = 16$. \qedhere
\end{proof}

\noindent The non-existence of a stable tree for $n=16$ directly implies that the \SNCG cannot have the finite improvement property, which states that every sequence of improving moves must be finite, i.e., reaches a Nash equilibrium. Thus, since the finite improvement property is equivalent to the game admitting a potential function~\cite{MS96} this implies the following statement.
\begin{corollary}
	The $\SNCG$ is not a potential game.
\end{corollary}

\begin{remark}
	By computational experiments we have obtained equilibrium trees for the $\SNCG$ for $1\leq n \leq 100$, except for $n=16$ and $n=18$. For $n=18$ we have verified via a brute-force search over all possible trees that no stable tree exists. Interestingly, for $n\geq 19$ equilibrium trees are no longer unique and in general the number of non-isomorphic equilibrium trees grows as $n$ grows.
\end{remark}

\subsection{Balanced Trees}\label{sec:balanced_trees}
Despite the negative result of the non-existence of a stable tree for $n=16$, in this section we prove the existence of NE's for infinitely many values of $n$. We prove this result by showing an interesting set of conditions for ruling out potential edge swaps; the proved conditions altogether allow us to show that there are infinitely many (balanced) trees that are stable. More precisely, we say that $T$ is {\em balanced} if any two nodes at the same distance from the root $r$ have equal in-degrees. Note, that any balanced tree $T$ of height $h$ can be uniquely encoded by a sequence of node degrees $(0,\dd_{h-1},\ldots,\dd_0)$, where $\dd_i$ is an in-degree of nodes at level $i$, i.e., at distance $i$ from the root. In this section we show that all the balanced trees of the form $(0,1,2,4,\dd_{h-4},\dots,\dd_0)$ such that $\dd_i < \dd_{i-1}\leq 2\dd_i +1$, for every $1 \leq i \leq h-4$ are stable. (See Fig.~\ref{fig:9_4_2_1_apx} for an example.)
\begin{figure}[h]
 \centering
 \includegraphics[width=\textwidth]{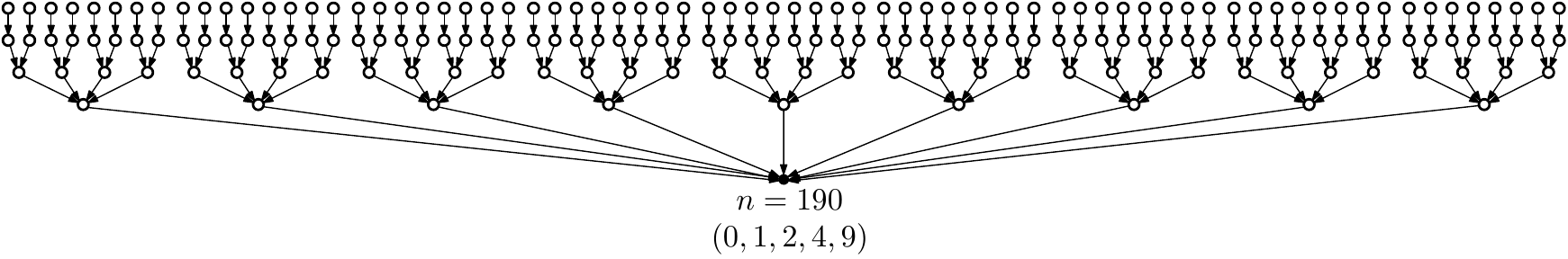}
 \caption{Sample of an extremal balanced tree with degree sequence $(0,1,2,4,9)$.}
 \label{fig:9_4_2_1_apx}
\end{figure}
Some of the provided conditions hold only for such balanced trees, while some other conditions hold for any tree (not necessarily balanced) that satisfies all the stated constraints.

The first two conditions (Condition 1 and Condition 2) rule out the case in which an agent swaps her edge towards her proper ancestors under the assumption that the tree satisfies some properties. In the rest of this section, we use the subscript $i$ to a node to emphasize that we are talking about a node that is at distance $i$ from the root of $T$.

\begin{lemma}[Condition 1]\label{lm:condition_1}
Consider a tree $T$ and a leaf-to-root path $P$ in $T$.
	Let $i \geq 2$. If $\dd^P_{i-2} \geq \dd^P_{i-1}$ and $|T^P_{i-1}| \geq \frac{\dd^P_{i-2}}{\dd^P_{i-2}+1-\dd^P_{i-1}}{|T^P_{i}|}$, then for $u^P_{i}$ it is not profitable to swap towards her ancestor of level $i-2$.
\end{lemma}

\begin{proof}
	Consider a tree $T$ for which the premises of the theorem are satisfied. If $u^P_i$ swaps towards her ancestor of level $i-2$, the $u_i$'s cost decreases by 
	\[
	\frac{\dd_{i-2}}{|T_{i-1}|}+\frac{\dd_{i-1}}{|T_{i}|} - \frac{\dd_{i-2}+1}{|T_{i}|}\leq \frac{\dd_{i-2}+1-\dd_{i-1}}{|T_i|}+\frac{\dd_{i-1}}{|T_{i}|} - \frac{\dd_{i-2}+1}{|T_{i}|}=0.
	\] 
	Hence, for $u^P_{i}$ it is not profitable to swap towards her ancestor of level $i-2$. 
\end{proof}

\begin{lemma}[Condition 2]\label{lm:condition_2}
Consider a tree $T$ and a leaf-to-root path $P$ in $T$.
	 Let $i \geq j+3$ and $u^P_i$ and $u^P_{j+2}$ two nodes such that $u^P_{j+2}$ is a proper ancestor of $u^P_i$. Furthermore, assume that:
	 \begin{enumerate}
	 	\item $|T^P_{j+2}| \geq 2|T^P_{i}|$;
	 	\item $\dd^P_{j} \geq \dd^P_{j+1}+1$;
	 	\item it is not profitable for $u^P_i$ to swap towards its ancestor of level $j+1$;
	 	\item it is not profitable for $u^P_{j+2}$ to swap towards its ancestor of level $j$.
	 \end{enumerate}
	 Then, it is not profitable for $u^P_i$  to swap towards her ancestor of level $j$.
\end{lemma}

\begin{proof}
	Let $A=\sum_{\ell=j+2}^{i-1}\frac{\dd_\ell}{|T_{\ell+1}|}$. Since for $u^P_i$ it is not profitable to swap towards her ancestor of level $j+1$, we have that $A+\frac{\dd_{j+1}}{|T_{j+2}|} \leq \frac{\dd_{j+1}+1}{|T_{i}|}$. 
	Similarly, since for $u^P_{j+2}$ it is not profitable to swap towards her ancestor of level $j$ we have that $\frac{\dd_{j+1}}{|T_{j+2}|}+\frac{\dd_j}{|T_{j+1}|} \leq  \frac{\dd_{j}+1}{|T_{j+2}|}$, i.e., $\frac{\dd_j}{|T_{j+1}|} \leq  \frac{\dd_{j}+1-\dd_{j+1}}{|T_{j+2}|}$. Let $T'$ be the tree obtained after the swap. Therefore,
	\begin{align*}
		cost_T(u_i) - \sum\limits_{k < j}\frac{\dd_{k}}{|T_{k+1}|} &= A+\frac{\dd_{j+1}}{|T_{j+2}|}+\frac{\dd_j}{|T_{j+1}|} \leq 
		\frac{\dd_{j+1}+1}{|T_{i}|} + \frac{\dd_{j}+1-\dd_{j+1}}{|T_{j+2}|}\\
		& \leq \frac{\dd_{j+1}+1}{|T_{i}|}+\frac{\dd_j+1-\dd_{j+1}}{2|T_{i}|} = \frac{2\dd_{j+1}+2+\dd_j+1-\dd_{j+1}}{2|T_{i}|}\\
		& =\frac{\dd_{j+1}+\dd_j+3}{2|T_{i}|} \leq \frac{2\dd_j+2}{2|T_{i}|} = \frac{\dd_{j}+1}{|T_{i}|}\\
		& = cost_{T'}(u_i)- \sum\limits_{k < j}\frac{\dd_{k}}{|T_{k+1}|}.
	\end{align*} 
	Hence, for $u^P_i$ it is not profitable to swap towards her ancestor of level $j$.
\end{proof}

\noindent The next four conditions (Condition 3, Condition 4, Condition 5, and Condition 6) rule out the case in which an agent $v_i$ swaps her edge towards some internal node $u_j$, with $u_j$ being neither a proper ancestor of $v_i$, nor a sibling of $v_i$. All the four conditions capture trees with some specific structure; furthermore, Condition 6 is suitable tailored for balanced trees $T$.

\begin{lemma}[Condition 3]\label{lm:condition_3}
	Let $v^{P_v}_i\in P_v$ and $u^{P_u}_j\in P_u$ be two nodes from two distinct leaf-to-root paths $P_v, P_u$ in $T$ such that $u^{P_u}_{j-1}$, the child of $u^{P_u}_j$, is not in the same branch of $v^{P_v}_i$. If for $v^{P_v}_i$ it is not profitable to swap towards $u^{P_u}_{j-1}$ and $|T^{P_v}_{i}| \geq \frac{\dd^{P_u}_{j-1}-\dd^{P_u}_j}{\dd^{P_u}_j}|T^{P_u}_j|$, then for $v^{P_v}_i$ it is not profitable to swap towards $u^{P_u}_j$.
\end{lemma}
\begin{proof}
	Let $k$ be the level of the intersection of two paths $P_v$ and $P_u$. 
	Let
	\begin{equation*}
		A=
		\begin{cases}
			\sum\limits_{\ell=k}^{j-2} \frac{\dd^{P_u}_\ell}{|T^{P_u}_{\ell+1}|+|T^{P_v}_{i}|}	&	\text{if $i \geq k+2$;}\\
			\frac{\dd_{k}-1}{|T_{k+1}|+|T^{P_v}_{i}|}+\sum\limits_{\ell=k+1}^{j-2} \frac{\dd^{P_u}_\ell}{|T^{P_u}_{\ell+1}|+|T^{P_v}_{i}|}	&	\text{otherwise (i.e., $i=k+1$).}
		\end{cases}
	\end{equation*}
	Since for $v^{P_v}_i$ it is not profitable to swap towards $u^{P_u}_{j-1}$, we have that the cost incurred by $v^{P_u}_i$ in $T$ is at most $	A+\frac{\dd_{j-1}+1}{|T^v_{i}|}.$
	
	From $|T^{P_v}_{i}| \geq \frac{\dd^{P_u}_{j-1}-\dd^{P_u}_j}{\dd^{P_u}_j}|T^{P_u}_{j}|$ we derive $\left(\dd^{P_u}_{j-1}-\dd^{P_u}_j\right)\left(|T^{P_u}_{j}|+|T^{P_v}_{i}|\right)\leq \dd^{P_u}_{j-1}\cdot|T^{P_v}_{i}|$, i.e.,
	\[
	\frac{\dd^{P_u}_{j-1}}{|T_{i}^{P_v}|} \leq \frac{\dd^{P_u}_{j-1}}{|T_{j}^{P_u}|+|T^{P_v}_{i}|}
	+\frac{\dd^{P_u}_j}{|T_{i}^{P_v}|}.
	\]
	As a consequence, the cost incurred by $v_i$ in $T$ is at most
	\[
	A+\frac{\dd^{P_u}_{j-1}+1}{|T_{i}^{P_v}|} = A + \frac{\dd^{P_u}_{j-1}}{|T_{i}^{P_v}|}+\frac{1}{|T_{i}^{P_v}|} \leq  A + \frac{\dd^{P_u}_{j-1}}{|T_{j}^{P_u}|+|T_{i}^{P_v}|}+\frac{\dd^{P_u}_j}{|T_{i}^{P_v}|} +\frac{1}{|T_{i}^{P_v}|} = A + \frac{\dd^{P_u}_{j-1}}{|T_{j}^{P_u}|+|T_{i}^{P_v}|}+\frac{\dd^{P_u}_j+1}{|T_{i}^{P_v}|}.
	\]
	Hence, for $v^{P_v}_i$ it is not profitable to swap towards $u^{P_u}_j$.
\end{proof}

\begin{lemma}[Condition 4]\label{lm:condition_4}
	Let $v^{P_v}_i\in P_v$ and $u^{P_u}_j\in P_u$ be two nodes from two distinct leaf-to-root paths $P_v, P_u$ in $T$ such that $u^{P_u}_{j-1}$, the child of $u^{P_u}_j$, is not in the same branch of $v^{P_v}_i$. If for $v^{P_v}_i$ it is not profitable to swap towards $u^{P_u}_{j}$ and $|T_{i}^{P_v}| \leq \frac{\dd_{j-1}^{P_u}-\dd^{P_u}_j}{\dd^{P_u}_j}|T_{j}^{P_u}|$, then for $v^{P_v}_i$ it is not profitable to swap towards $u^{P_u}_{j-1}$.
\end{lemma}
\begin{proof}
	Let $v^{P_v}_k=u^{P_u}_k$ be the least common ancestor of $v^{P_v}_i$ and $u^{P_u}_j$. Let
	\begin{equation*}
		A=
		\begin{cases}
			\sum\limits_{\ell=k}^{j-2} \frac{\dd^{P_u}_\ell}{|T_{\ell+1}^{P_u}|+|T_{i}^{P_v}|}	&	\text{if $i \geq k+2$;}\\
			\frac{\dd_{k}-1}{|T_{k+1}^{P_u}|+|T_{i}^{P_v}|}+\sum\limits_{\ell=k+1}^{j-2} \frac{\dd_\ell^{P_u}}{|T_{\ell+1}^{P_u}|+|T_{i}^{P_v}|}	&	\text{otherwise (i.e., $i=k+1$).}
		\end{cases}
	\end{equation*}
	Since for $v^{P_v}_i$ it is not profitable to swap towards $u^{P_u}_{j}$, we have that the cost incurred by $v^{P_v}_i$ in $T$ is at most 
	\[
	A + \frac{\dd_{j-1}}{|T_{j}^{P_u}|+|T_{i}^{P_v}|}+\frac{\dd^{P_u}_j+1}{|T_{i}^{P_v}|}.
\]
	From $|T_{i}^{P_v}| \leq \frac{\dd_{j-1}-\dd_j}{\dd_j}|T_{j}^{P_u}|$ we derive $\left(\dd^{P_u}_{j-1}-\dd_j^{P_u}\right)\left(|T_{j}^{P_u}|+|T_{i}^{P_v}|\right)\geq \dd^{P_u}_{j-1}|T_{i}^{P_u}|$, i.e.,
		\[
	\frac{\dd_{j-1}^{P_u}}{|T_{i}^{P_v}|} \geq \frac{\dd^{P_u}_{j-1}}{|T_{j}^{P_u}|+|T^{P_v}_{i}|}+\frac{\dd_j}{|T_{i}^{P_v}|}.
	\]
	As a consequence, the cost incurred by $v^{P_v}_i$ in $T$ is at most
		\[
	A + \frac{\dd^{P_u}_{j-1}}{|T_{j}^{P_u}|+|T_{i}^{P_v}|}+\frac{\dd^{P_u}_j+1}{|T_{i}^{P_v}|} = 
	A + \frac{\dd^{P_u}_{j-1}}{|T^{P_u}_{j}|+|T^{P_v}_{i}|}+\frac{\dd^{P_u}_j}{|T_{i}^{P_v}|} +\frac{1}{|T_{i}^{P_v}|} \leq 
	A + \frac{\dd^{P_u}_{j-1}}{|T_{i}^{P_v}|}+\frac{1}{|T_{i}^{P_v}|}
	\leq A+\frac{\dd^{P_u}_{j-1}+1}{|T^{P_v}_{i}|}.
		\]
	Hence, for $v^{P_v}_i$ it is not profitable to swap towards $u^{P_u}_{j-1}$.
\end{proof}

\begin{lemma}[Condition 5]\label{lm:condition_5}
	Let $T$ be a balanced tree of height $h$ such that $\dd_j \leq 2\dd_{j+1}+1$ for every $j<h$. Let $v_i$ and $u_{i}$ be two sibling nodes. Then, it is not profitable for $v_i$ to swap towards $u_i$.
\end{lemma}
\begin{proof} Follows directly from Lemma~\ref{thm:deg_up}.
\end{proof}

\noindent The next corollary, which holds only for some balanced trees, is implied by Lemma~\ref{lm:condition_3} and Lemma~\ref{lm:condition_4}.

\begin{corollary}\label{cor:balanced_trees_lateral_swaps_towards_internal_vertices}
	Let $T$ be a balanced tree of height $h$ such that $\dd_{i+1} < \dd_i \leq 2\dd_{i+1}$ for every $i \leq h-3$. Let $v_i$ and $u_{i-1}$ be such that $u_{i-1}$ is not an ancestor of $v_i$ in $T$. Let $k$ be the distance from the root to the least common ancestor of $v_i$ and $u_{i-1}$. If it is not profitable for $v_i$ to swap towards $u_{i-1}$, then it is not profitable for $v_i$ to swap towards $u_j$  for every $j > k$ with $\dd_j > 0$.
\end{corollary}
\begin{proof}
	We divide the proof into two complementary cases, according to whether $j\geq i$ or not.
	In the former case, for every $j \geq i$, we have that
	\[
	|T_{i}| = \frac{2\dd_j-\dd_j}{\dd_j} |T_{i}| \geq  \frac{2\dd_j-\dd_j}{\dd_j} |T_{j}| \geq  \frac{\dd_{j-1}-\dd_j}{\dd_j} |T_{j}|.
	\]
	Therefore, thanks to Lemma~\ref{lm:condition_3}, it is not profitable for $v_i$ to swap towards $u_j$ for every $j \geq i$.
	In the latter case, for every $k+1 < j \leq i-1$, we have that
	\[
	|T_{i}| < \frac{1}{\dd_{i-1}} |T_{i-1}| \leq \frac{\dd_{j-1}-\dd_j}{\dd_j}|T_{j}|.
	\]
	Therefore, thanks to Lemma~\ref{lm:condition_4}, it is not profitable for $v_i$ to swap towards $u_{j-1}$.
\end{proof}

\noindent The next corollary, similar to Corollary \ref{cor:balanced_trees_lateral_swaps_towards_internal_vertices} allows us to capture the case of extremal balanced trees that are stable.
\begin{corollary}\label{cor:extremal_balanced_trees_lateral_swaps_towards_internal_vertices}
	Let $T$ be a balanced tree of height $h$ such that $\dd_{i+1} < \dd_i \leq 2\dd_{i+1}+1$ for every $i \leq h-3$. Let $v_i$ and $u_{i-1}$ be such that $u_{i-1}$ is not an ancestor of $v_i$ in $T$. Let $k$ be the distance from the root to the least common ancestor of $v_i$ and $u_{i-1}$. 
	If it is not profitable for $v_i$ to swap towards $u_{i-1}$ and it is not profitable for $v_i$ to swap towards $u_i$, then it is not profitable for $v_i$ to swap towards $u_j$ for every $j > k$ with $\dd_j > 0$.
\end{corollary}
\begin{proof}
In case $j \geq i+1$, we have that
	\[
	|T_{i}| > \dd_{j-1} |T_{j}| > \frac{\dd_{j-1}-\dd_{j}}{\dd_{j}} |T_{j}|.
	\]
	Therefore, thanks to Lemma~\ref{lm:condition_3}, it is not profitable for $v_i$ to swap towards $u_j$ for every $j \geq i+1$.
	From the other hand, for every $k+1 < j \leq i-1$, we have that
	\[
	|T_{i}|\leq (\dd_{j-1}-\dd_j)|T_{i}|< \frac{\dd_{j-1}-\dd_j}{\dd_j}|T_{j}|.
	\]
	Therefore, thanks to Lemma~\ref{lm:condition_4}, it is not profitable for $v_i$ to swap towards $u_{j-1}$.
\end{proof}

\begin{lemma}[Condition 6]\label{lm:balanced_trees_lateral_swap_towards_non_sibling_vertex_same_level_minus_one}
	Let $T$ be a balanced tree of height $h$ such that $\dd_{h-1}=1$, $\dd_{h-2}=2$, $\dd_{h-3}=4$, and $\dd_{j+1} \leq \dd_j \leq 2\dd_{j+1}+1$ for every $j \leq h-4$. Let $v_i$ and $u_{i-1}$ be such that $u_{i-1}$ is not an ancestor of $v_i$ in $T$. Then, it is not profitable for $v_i$ to swap towards $u_{i-1}$.
\end{lemma}
\begin{proof}
	Let $k$ be the distance from the root with respect to the least common ancestor of $v_i$ and $u_{i-1}$. We have to prove that
	\[
	\frac{\dd_{i-1}}{|T_{i}|}+\sum_{\ell=k}^{i-2}\frac{\dd_\ell}{|T_{\ell+1}|} \leq \frac{\dd_{i-1}+1}{|T_{i}|}+\sum_{\ell=k}^{i-2}\frac{\dd_\ell}{|T_{\ell+1}|+|T_{i}|}, \text{ i.e., } \sum_{\ell=k}^{i-2}\frac{\dd_\ell |T_{i}|^2}{|T_{\ell+1}|(|T_{\ell+1}|+|T_{i}|)} \leq 1.
	\]
	We prove the last inequality by showing that $\frac{\dd_\ell |T_{i}|^2}{|T_{\ell+1}|(|T_{\ell+1}|+|T_{i}|)} \leq \frac{1}{2^{i-1-\ell}}$ for every $\ell \leq i-2$.
	The proof is by induction on $\ell$.                                                                                                                                                                                                                 
	We prove the base case $\ell=i-2$ first. 
	The proof is by cases. 
	
	When $i=h$, we have that $\dd_{i-1}=1$, $|T_{i}|=1$, $\dd_{i-2}=2$, and $|T_{i-1}|=2$; as a consequence $\frac{\dd_{i-2} |T_{i}|^2}{|T_{i-1}|(|T_{i-1}|+|T_{i}|)} =\frac{2}{2(2+1)} < \frac{1}{2}$ and the claim follows. 
	
	When $i=h-1$, we have that $\dd_{i-1}=2$, $|T_{i}|=2$, $\dd_{i-2}=4$, and $|T_{i-1}|=5$; therefore, $\frac{\dd_{i-2} |T_{i}|^2}{|T_{i-1}|(|T_{i-1}|+|T_{i}|)}=\frac{4\cdot 2^2}{5(5+2)} < \frac{1}{2}$, and the claim follows. 
	
	When $i \leq h-2$, we have that $\dd_{i-1}\geq 4$; as a consequence, using also the facts that $|T_{i-1}|>\dd_{i-1}|T_{i}|$ and $\dd_{i-2}\leq 2\dd_{i-1}+1$, we have that
	$\frac{\dd_{i-2} |T_{i}|^2}{|T_{i-1}|(|T_{i-1}|+|T_{i}|)} < \frac{(2\dd_{i-1}+1) |T_{i}|^2}{\dd_{i-1}|T_{i}|(\dd_{i-1}+1)|T_{i}|} < \frac{2(\dd_{i-1}+1)}{4(\dd_{i-1}+1)} \leq \frac{1}{2}$, and the claim follows.
	
	To prove the inductive case, it is enough to show that $\frac{\dd_{j-1}}{|T_{j}|(|T_{j}|+|T_{i}|)} \leq \frac{\dd_{j}}{2|T_{j+1}|(|T_{j+1}|+|T_{i}|)}$ for every $j \leq i-2$. 
	Indeed, if by induction we assume $\frac{\dd_\ell |T_{i}|^2}{|T_{\ell+1}|(|T_{\ell+1}+|T_{i}|)} \leq \frac{1}{2^{i-1-\ell}}$, then 
%	\[
$	\frac{\dd_{\ell-1} |T_{i}|^2}{|T_{\ell}|(|T_{\ell}|+|T_{i}|)} \leq \frac{\dd_\ell |T_{i}|^2}{2|T_{\ell+1}|(|T_{\ell+1}|+|T_{i}|)} \leq \frac{1}{2^{i-1-(\ell-1)}},$
thus completing the proof. 
%	\] 
	
	We prove $\frac{\dd_{j-1}}{|T_{j}|(|T_{j}|+|T_{i}|)} \leq \frac{\dd_{j}}{2|T_{j+1}|(|T_{j+1}|+|T_{i}|)}$ by showing that $2\dd_{j-1} |T_{j+1}|(|T_{j+1}|+|T_{i}|) \leq \dd_{j}|T_{j}|(|T_{j}|+|T_{i}|)$. 
	We claim that the last inequality holds under the assumption that $$5 |T_{j+1}|(|T_{j+1}|+|T_{i}|)\leq |T_{j}|(|T_{j}|+|T_{i}|). $$ 
	Indeed, if we assume $5 |T_{j+1}|(|T_{j+1}|+|T_{i}|)\leq |T_{j}|(|T_{j}|+|T_{i}|) $, since $\dd_j \geq 2$ and $\dd_{j-1} \leq 2\dd_j+1$, and using the fact that $|T_{j}| > \dd_j|T_{j+1}|$, we obtain
	\begin{align*}
	2\dd_{j-1} |T_{j+1}|(|T_{j+1}|+|T_{i}|) &\leq 2(2\dd_{j}+1)|T_{j+1}|(|T_{j+1}|+|T_{i}|)\\
	& \leq 5\dd_j |T_{j+1}|(|T_{j+1}|+|T_{i}|)\leq \dd_{j} |T_{j}|(|T_{j}|+|T_{i}|).
	\end{align*}
	So, it remains to prove that $|T_{j}|(|T_{j}|+|T_{i}|) \geq 5 |T_{j+1}|(|T_{j+1}|+|T_{i}|)$. 
	The proof is by cases. 
	For the case in which $j-1 \leq h-4$, we have that $|T_{j}| \geq 4|T_{j+1}|$ and therefore, $|T_{j}|(|T_{j}|+|T_{i}|) \geq 5 |T_{j+1}|(|T_{j+1}|+|T_{i}|)$. 
	For the case in which $j-1=h-3$, i.e., $j=h-2$, we have that $i=h$, $|T_{h}|=1$, $|T_{h-1}|=2$, and $|T_{h-2}|=5$; as a consequence $|T_{j}|(|T_{j}|+|T_{i}|)=5(5+1)=30=5\cdot 2(2+1) = 5|T_{j+1}|(|T_{j+1}|+|T_{i}|)$. The claim follows.
\end{proof}

\noindent The following lemma allows us to show that also {\em extremal} balanced trees, i.e., balanced trees in which the in-degree of each vertex is set to the maximum possible value, is stable.

\begin{lemma}\label{lm:balanced_trees_lateral_swap_towards_non_sibling_vertex_same_level}
	Let $T$ be a balanced tree of height $h$ such that $\dd_{h-1}=1$, $\dd_{h-2}=2$, $\dd_{h-3}=4$, and $\dd_{j-1} \leq \dd_j \leq 2\dd_{j+1}+1$ for every $j \leq h-4$. Let $v_i$ and $u_i$ be two distinct nodes such that $v_i$ and $u_i$ are not siblings. Then, it is not profitable for $v_i$ to swap towards $u_i$.
\end{lemma}
\begin{proof}
	Let $k$ be the distance from the root with respect to the least common ancestor of $v_i$ and $u_i$. We have to prove that
	\[
	\sum_{\ell=k}^{i-1}\frac{\dd_\ell}{|T_{\ell+1}|} \leq \frac{\dd_i+1}{|T_{i}|}+\sum_{\ell=k}^{i-1}\frac{\dd_\ell}{|T_{\ell+1}|+|T_{i}|}, \text{ i.e., }
	\sum_{\ell=k}^{i-2}\frac{\dd_\ell |T_{i}|^2}{|T_{\ell+1}|(|T_{\ell+1}|+|T_{i}|)} \leq \frac{2\dd_i+2-\dd_{i-1}}{2}.
	\]
	When $\dd_{i-1}=2\dd_i$ we have that $2\dd_i+2-\dd_{i-1}=1$ and therefore, in this case, we have to show that
	$\sum_{\ell=k}^{i-2}\frac{\dd_\ell |T_{i}|^2}{|T_{\ell+1}|(|T_{\ell+1}|+|T_{i}|)} \leq 1$. 
	We observe that such a proof has already been provided in Lemma~\ref{lm:balanced_trees_lateral_swap_towards_non_sibling_vertex_same_level_minus_one}. 
	Furthermore, we have that $\dd_{i-1}=2\dd_i$ when $i=h-1,h-2$. Therefore, it remains to prove the claim when $i \leq h-3$ and $\dd_{i-1}=2\dd_{i}+1$.
	
	When $i \leq h-3$ and $\dd_{i-1}\leq 2\dd_{i}+1$, we have that $2\dd_i+2-\dd_{i-1} \geq 1$. 
	Therefore, to complete the proof for the case $i \leq h-3$, it is enough to show that $\sum_{\ell=k}^{i-2}\frac{\dd_\ell |T_{i}|^2}{|T_{\ell+1}|(|T_{\ell+1}|+|T_{i}|)} \leq \frac{1}{2}$.	
	We prove that such an inequality holds by showing that $\frac{\dd_\ell |T_{i}|^2}{|T_{\ell+1}|(|T_{\ell+1}|+|T_{i}|)} \leq \frac{1}{2^{i-\ell}}$ for every $\ell \leq i-2$. The proof is by induction on~$\ell$.                                                                                                                                                           
	
	We prove the base case $\ell=i-2$ first. Since $i \leq h-3$ and $\dd_{i-1} \geq 9$, using also the facts that $|T_{i-1}|>\dd_{i-1}|T_{i}|$ and $\dd_{i-2}\leq 2\dd_{i-1}+1$, we have that
\[
	\frac{\dd_{i-2} |T_{i}|^2}{|T_{i-1}|(|T_{i-1}|+|T_{i}|)} \leq \frac{(2\dd_{i-1}+1)|T_{i}|^2}{\dd_{i-1}|T_{i}|(\dd_{i-1}+1)|T_{i}|} =\frac{2\dd_{i-1}+1}{\dd_{i-1}(\dd_{i-1}+1)} \leq \frac{1}{4}.
\]

\noindent	To prove the inductive case, it is enough to show that $\frac{\dd_{j-1}}{|T_{j}|(|T_{j}|+|T_{i}|)} \leq \frac{\dd_{j}}{2|T_{j+1}|(|T_{j+1}|+|T_{i}|)}$ for every $j \leq i-2$. Indeed, if by induction we assume $\frac{\dd_\ell |T_{i}|^2}{|T_{\ell+1}|(|T_{\ell+1}|+|T_{i}|)} \leq \frac{1}{2^{i-\ell}}$, then 
	\[
	\frac{\dd_{\ell-1} |T_{i}|^2}{|T_{\ell}|(|T_{\ell}|+|T_{i}|)} \leq \frac{\dd_\ell |T_{i}|^2}{2|T_{\ell+1}|(|T_{\ell+1}|+|T_{i}|)} \leq \frac{1}{2^{i-(\ell-1)}},
	\] 
	thus completing the proof. We observe that $\frac{\dd_{j-1}}{|T_{j}|(|T_{j}|+|T_{i}|)} \leq \frac{\dd_{j}}{2|T_{j+1}|(|T_{j+1}|+|T_{i}|)}$ has already been proved in Lemma~\ref{lm:balanced_trees_lateral_swap_towards_non_sibling_vertex_same_level_minus_one}. The claim follows.
\end{proof}

\noindent The next three lemmas, the first of which holds only for non-extremal balanced trees, rule out the case in which $v_i$ swaps her edge towards a leaf.

\begin{lemma}\label{lm:balanced_tree_lateral_swap_heavy_subtree_towards_leaf}
	Let $T$ be a balanced tree of height $h$ such that $\dd_{h-1}=1$, $\dd_{h-2}=2$, $\dd_{h-3}=4$, and $\dd_{j+1} \leq \dd_j \leq 2\dd_{j+1}$ for every $j \leq h-4$. Then, for every $i \leq h-2$, it is not profitable for $v_i$ to swap towards $u_h$.
\end{lemma}
\begin{proof}
	Clearly, we only need to consider the case in which $u_h$ is not a descendant of $v_i$. Let $k$ be the least common ancestor of $v_i$ and $u_{h-2}$, and denote by $A+\frac{3}{|T_{i}|}$ the cost incurred by $v_i$ when $v_i$ swaps towards $u_{h-2}$ ($A$ may be equal to 0). 
	Since by Lemma~\ref{lm:condition_5} and Lemma~\ref{lm:balanced_tree_lateral_swap_heavy_subtree_towards_leaf}, it is not profitable for $v_i$ to swap towards $u_{h-2}$, the cost of $v_i$ in $T$ is at most $A+\frac{3}{|T_{i}|}$. Furthermore, since $|T_{i}| \geq 5$, we have that  $\frac{3}{|T_{i}|} \leq \frac{2}{|T_{i}|+2}+\frac{1}{|T_{i}|+1}+\frac{1}{|T_{i}|}$. Therefore, the cost incurred by $v_i$ in $T$ is at most $A+\frac{2}{|T_{i}|+2}+\frac{1}{|T_{i}|+1}+\frac{1}{|T_{i}|}$. But this is exactly the cost incurred by $v_i$ if she swapped her edge towards $u_h$. The claim follows.
\end{proof}

\begin{lemma}\label{lm:balanced_tree_lateral_swap_of_leaf_towards_leaf}
	Let $T$ be a balanced tree of height $h$ such that $\dd_{h-1}=1$, $\dd_{h-2}=2$, $\dd_{h-3}=4$, and $\dd_{j+1} \leq \dd_j \leq 2\dd_{j+1}+1$ for every $j \leq h-4$. Let $v_h$ and $u_h$ be two distinct leaves of $T$. Then, it is not profitable for $v_h$ to swap towards $u_h$.
\end{lemma}
\begin{proof}
	Let $k$ be the distance from the root with respect to the least common ancestor of $u_h$ and $v_h$. We have to prove that
	\[
	\sum_{\ell=k}^{h-1} \frac{\dd_\ell}{|T_{\ell+1}|} \leq 1+\sum_{\ell=k}^{h-1}\frac{\dd_\ell}{|T_{\ell+1}|+1}, \text{ i.e., }
	\sum_{\ell=k}^{h-1} \frac{\dd_\ell}{|T_{\ell+1}|(|T_{\ell+1}|+1)} \leq 1.
	\]
	We prove the last inequality by showing that (a) $\sum\limits_{\ell=h-4}^{h-1} \frac{\dd_\ell}{|T_{\ell+1}|(|T_{\ell+1}|+1)} \leq \frac{76}{77}$ and (b) $\frac{\dd_\ell}{|T_{\ell+1}|(|T_{\ell+1}|+1)} \leq \frac{1}{77\cdot 2^{h-4+\ell}}$ for every $\ell \leq h-5$.
	We prove (a) first. Using all the hypothesis, we have that 
	\[
	\sum_{\ell=h-4}^{h-1} \frac{\dd_\ell}{|T_{\ell+1}|(|T_{\ell+1}|+1)} = \frac{1}{2}+\frac{1}{3}+\frac{2}{15}+\frac{\dd_{h-4}}{462} \leq  \frac{29}{30}+\frac{9}{462} < \frac{76}{77}.
	\]
	We now prove (b). The proof is by induction on $\ell$. For the base case $\ell=h-5$, we have that 
	\[
	\frac{\dd_{h-5}}{|T_{h-4}|(|T_{h-4}|+1)} \leq \frac{2\dd_{h-4}}{|T_{h-4}|^2}< \frac{2\dd_{h-4}}{|T_{h-3}|^2\dd_{h-4}^2} = \frac{2}{21^2\cdot  4}=\frac{1}{882}<\frac{1}{144},
	\]
	and the claim follows.
	
	To prove the inductive case, it is enough to show that $\frac{\dd_{j-1}}{|T_{j}|(|T_{j}|+1)} \leq \frac{\dd_{j}}{2|T_{j+1}|(|T_{j+1}|+1)}$ for every $j \leq h-5$. Indeed, if by induction we assume $\frac{\dd_\ell}{|T_{\ell+1}|(|T_{\ell+1}|+1)} \leq \frac{1}{77\cdot 2^{h-4-\ell}}$, then 
	\[
	\frac{\dd_{\ell-1}}{|T_{\ell}|(|T_{\ell}|+1)} \leq \frac{\dd_\ell}{2|T_{\ell+1}|(|T_{\ell+1}|+1)} \leq \frac{1}{77\cdot 2^{h-4-(\ell-1)}},
	\] 
	thus completing the proof. We prove $\frac{\dd_{j-1}}{|T_{j}|(|T_{j}|+1)} \leq \frac{\dd_{j}}{2|T_{j+1}|(|T_{j+1}|+1)}$ by showing that $2\dd_{j-1}\cdot |T_{j+1}|(|T_{j+1}|+1) \leq \dd_{j}|T_{j}|(|T_{j}|+1)$. Since $\dd_j,|T_{j+1}| \geq 4$ and $|T_{j}|>4|T_{j+1}|$, we have that
	\[
	2\dd_{j-1} |T_{j+1}|(|T_{j+1}|+1) \leq 2(2\dd_{j}+1)|T_{j+1}|(|T_{j+1}|+1) \leq 5\dd_j|T_{j+1}|(|T_{j+1}|+1)\leq \dd_{j} |T_{j}|(|T_{j}|+1). \qedhere
	\]	
\end{proof}

\begin{lemma}\label{lm:balanced_tree_lateral_swap_of_light_subtree_towards_leaf}
	Let $T$ be a balanced tree of height $h$ such that $\dd_{h-1}=1$, $\dd_{h-2}=2$, $\dd_{h-3}=4$, and $\dd_{j+1} \leq \dd_j \leq 2\dd_{j+1}+1$ for every $j \leq h-4$. Let $v_{h-1}$ and $u_h$ be two nodes of $T$ such that $u_h$ is not a child of $v_{h-1}$. Then, it is not profitable for $v_{h-1}$ to swap towards $u_h$.
\end{lemma}
\begin{proof}
	Let $k$ be the distance from the root with respect to the least common ancestor of $u_h$ and $v_{h-1}$. We have to prove that
	\[
	\sum_{\ell=k}^{h-2} \frac{\dd_\ell}{|T_{\ell+1}|} \leq \frac{1}{2}+\frac{1}{3}+\sum_{\ell=k}^{h-2}\frac{\dd_\ell}{|T_{\ell+1}|+2}, \text{ i.e., }
	\sum_{\ell=k}^{h-2} \frac{2\dd_\ell}{|T_{\ell+1}|(|T_{\ell+1}|+2)} \leq \frac{5}{6}.
	\]
	We prove the last inequality by showing that (a) $\sum\limits_{\ell=h-4}^{h-2} \frac{\dd_\ell}{|T_{\ell+1}|(|T_{\ell+1}|+2)} \leq \frac{19}{24}$ and (b) $\frac{\dd_\ell}{|T_{\ell+1}|(|T_{\ell+1}|+2)} \leq \frac{1}{24\cdot 2^{h-4+\ell}}$ for every $\ell \leq h-5$.\footnote{We observe that $\frac{1}{24}=\frac{5}{6}-\frac{19}{24}$.}
	We prove (a) first. Using all the hypothesis, we have that 
	\[
	\sum_{\ell=h-4}^{h-2} \frac{2\dd_\ell}{|T_{\ell+1}|(|T_{\ell+1}|+2)}= \frac{2\cdot 2}{2(2+2)}+\frac{2\cdot 4}{5(5+2)}+\frac{2\dd_{h-4}}{21(21+2)} \leq  \frac{1}{2}+\frac{8}{35}+\frac{18}{483} < \frac{19}{24}.
	\]
	We now prove (b). The proof is by induction on $\ell$. For the base case $\ell=h-5$, we have that 
	\[
	\frac{2\dd_{h-5}}{|T_{h-4}|(|T_{h-4}|+1)} \leq \frac{4\dd_{h-4}}{|T_{h-4}|^2}< \frac{4\dd_{h-4}}{|T_{h-3}|^2\dd_{h-4}^2} = \frac{4}{21^2\cdot  4}=\frac{1}{441}<\frac{1}{24\cdot 2},
	\]
	and the claim follows.
	
	To prove the inductive case, it is enough to show that $\frac{\dd_{j-1}}{|T_{j}|(|T_{j}|+2)} \leq \frac{\dd_{j}}{2|T_{j+1}|(|T_{j+1}|+2)}$ for every $j \leq h-5$. Indeed, if by induction we assume $\frac{2\dd_\ell}{|T_{\ell+1}|(|T_{\ell+1}|+2)} \leq \frac{1}{24\cdot 2^{h-4-\ell}}$, then 
	\[
	\frac{2\dd_{\ell-1}}{|T_{\ell}|(|T_{\ell}|+2)} \leq \frac{2\dd_\ell}{2|T_{\ell+1}|(|T_{\ell+1}|+2)} \leq \frac{1}{24\cdot 2^{h-4-(\ell-1)}},
	\] 
	thus completing the proof. We prove $\frac{\dd_{j-1}}{|T_{j}|(|T_{j}|+2)} \leq \frac{\dd_{j}}{2|T_{j+1}|(|T_{j+1}|+2)}$ by showing that $2\dd_{j-1} |T_{j+1}|\cdot(|T_{j+1}|+2) \leq \dd_{j}|T_{j}|(|T_{j}|+2)$. Since $\dd_j,|T_{j+1}| \geq 4$ and $|T_{j}|>4|T_{j+1}|$, we have that
	\[
	2\dd_{j-1} |T_{j+1}|(|T_{j+1}|+2) \leq 2(2\dd_{j}+1)|T_{j+1}|(|T_{j+1}|+2) \leq 5\dd_j|T_{j+1}|(|T_{j+1}|+2)\leq \dd_{j} |T_{j}|(|T_{j}|+2). \qedhere
	\]	
\end{proof}

\begin{theorem}\label{thm:some_balanced_trees_are_NE}
The balanced tree $T$ with degree sequence $(0,1,2,4,\dd_{h-4},\dots,\dd_0)$, where $\dd_{j+1} < \dd_j \leq 2\dd_{j+1}+1$ for every $j \leq h-4$, is stable.
\end{theorem}

\begin{proof}
	Let $v_i$ be a fixed node in $T$. We prove that $v_i$ is playing her best response in $T$. The proof is by cases.

	In the first case, we prove that it is no profitable for $v_i$ to swap her edge towards any leaf of $T$. Clearly, we only need to consider leaves that are not descendants of $v_i$. Lemma~\ref{lm:balanced_tree_lateral_swap_of_leaf_towards_leaf} covers the case $i=h$. Lemma~\ref{lm:balanced_tree_lateral_swap_of_light_subtree_towards_leaf} covers the case $i=h-1$. Lemma~\ref{lm:balanced_tree_lateral_swap_heavy_subtree_towards_leaf} covers the case $i\leq h-2$.
	
	In the second case, we prove that it is not profitable for $v_i$ to swap her edge towards any internal node of $T$. Clearly, we only need to consider internal nodes that are not descendants of $v_i$. We divide the proof into three cases. 
	
	We start proving that it is not profitable for $v_i$ to swap her edge towards any of its ancestor nodes~$v_j$, with $j \leq i-2$. The proof is by induction on $j$.
	The base case is when $j=i-2$. In Lemma~\ref{lm:condition_1}, we have proved that it is not profitable for any node $v_i$ to swap her edge towards its ancestor $v_{i-2}$ under the condition $|T_{i-1}|\geq\frac{\dd_{i-2}}{\dd_{i-2}+1-\dd_{i-1}}|T_{i}|$. 
	We prove that this condition holds for every $i$. First of all, we observe that $\dd_{i-2}+1-\dd_{i-1}\geq \dd_{i-1}+1+1-\dd_{i-1}=2$. Therefore,
	\[
	|T_{i-1}| > \dd_{i-1}|T_{i}| \geq \frac{\dd_{i-2}}{2} |T_{i}| \geq \frac{\dd_{i-2}}{\dd_{i-2}+1-\dd_{i-1}}|T_{i}|.
	\]
	Hence, it is not profitable for $v_i$ to swap towards its ancestor $v_{i-2}$.
	We prove the inductive case $j \leq i-2$, i.e., $i\geq j+3$. We observe that conditions (1-4) of Lemma~\ref{lm:condition_2} are all satisfied. Therefore, it is not profitable for $v_i$ to swap towards its ancestor $v_j$.
	
    Let $u_i$ be a sibling of $v_i$ in $T$ and let $u_j$, with $i \leq j < h$ be any descendant of $u_i$ in $T$ ($u_j$ may also be equal to $u_i$). We prove by induction on $j$ that it is not profitable for $v_i$ to swap towards $u_j$. 
    The base case $j=i$ has already been proved in Lemma~\ref{lm:condition_5}. For the inductive case $i<j$, we simply have to prove that the condition $|T_{i}|\geq \frac{\dd_{j-1}-\dd_j}{\dd_j}|T_{j}|$ of  Lemma~\ref{lm:condition_4} is satisfied. Such a condition holds; indeed, since $\dd_j \geq 1$ and $\dd_{j-1}\leq 2\dd_j$, we have that
\[
|T_{i}|\geq \dd_i|T_{i+1}| \geq \dd_{j-1} |T_{j}| \geq \frac{\dd_{j-1}-\dd_j}{\dd_j}|T_{j+1}|.
\]

	Let $u_{i-1}$ be a node of $T$ that is not an ancestor of $v_i$ in $T$, and let $k$ be the least common ancestor of $v_i$ and $u_{i-1}$. By Lemma~\ref{lm:balanced_trees_lateral_swap_towards_non_sibling_vertex_same_level_minus_one}, it is not profitable for $v_i$ to swap towards $u_{i-1}$. Furthermore, if $u_i$ is a child of $u_{i-1}$, then, by Lemma~\ref{lm:balanced_trees_lateral_swap_towards_non_sibling_vertex_same_level}, it is not profitable for $v_i$ to swap towards $u_{i-1}$. Finally, since all the conditions of Corollary~\ref{cor:balanced_trees_lateral_swaps_towards_internal_vertices} and Corollary~\ref{cor:extremal_balanced_trees_lateral_swaps_towards_internal_vertices} are satisfied, we have that it is not profitable for $v_i$ to swap towards $u_j$ for every $j > k$.
	
	The claim follows. 
\end{proof}

\noindent From Theorem~\ref{thm:some_balanced_trees_are_NE} we derive the following corollary.

\begin{corollary}\label{cor:NE_exists_for_inf_many_n}
The \SNCG with $n$ agents admits a NE for infinitely many values of $n \in \mathbb{N}$.
\end{corollary}

\noindent By Corollary~\ref{cor:NE_exists_for_inf_many_n}, we observe that NE exists for all $n$ that admit an existence of a balanced tree. Intuitively, a minor modification of a balanced tree, e.g., removing a subset of leaf nodes, keep the tree stable. Moreover, for $n\geq 19$ we have found several non-isomorphic equilibrium trees in each case. The number of non-isomorphic equilibria grows with $n$, which indicates that for growing $n$ also the number of possibilities how to combine suitable equilibrium trees into larger equilibrium trees grows.
Therefore, we conjecture the existence of stable trees for all values of $n$ except for $n=16$ and $n=18$. We believe that this conjecture can be proven by a dynamic programming approach that exploits the different possibilities of how equilibrium sub-trees can be combined into larger equilibrium trees.

\begin{conjecture}\label{conj:existence}
	For any $n \in \mathbb{N}$, with $n\neq 16$ and $n\neq 18$ a pure NE exists in the \SNCG. 
\end{conjecture}

\section{Quality of Equilibrium Trees}\label{sec:quality_of_NE}
In this section we provide results on the quality of stable networks. In particular, we prove a constant upper bound on the PoA and give lower bounds on the PoA and PoS. Furthermore, we prove an upper bound on the PoS for certain balanced trees.
We first observe that any network in which at least one node has in-degree 2 is not a social optimum. Hence, a Hamilton path is the social optimum. 
\begin{theorem}\label{thm:opt_is_a_path}
	Any Hamiltonian path having the root $r$ as one endnode is a social optimum.
\end{theorem}

\begin{proof}
	Since any agent buys exactly one edge and any edge $\left(u,v\right)$ has cost of at least $1$, since $\indeg_T(v) \geq 1$ because of the edge itself, the social cost, i.e., the overall sum of the costs incurred by all the agents, of any solution is at least $n$. Any Hamiltonian path having the root $r$ as one endnode is a network whose social cost is exactly equal to $n$ since every edge has cost of exactly $1$. As a consequence, such a Hamiltonian path is a social optimum. Any other network in which at least one node has in-degree 2 is not a social optimum as the cost of such a network is at least $n+2$ since the cost of two edges is at least $2$ and the cost of the other $n-2$ edges is at least~$1$.
\end{proof}

\subsection{Price of Anarchy}\label{sec:PoA}
In every network $T$ for all $v \in V$, $\indeg_T(v) \leq n$, since there are exactly $n$ edges. Hence, the cost of an agent is upper bounded by $n$ and the star graph yields a trivial upper bound of $n$ for the PoA. However, we prove next a constant upper bound on the PoA.

\begin{theorem}\label{thm:PoA}
	The PoA is at most 8.62.
\end{theorem}
\begin{proof}
	Consider a stable network $T=(V,E)$.
	By Theorem~\ref{thm:opt_is_a_path},  the social optimum is a path of cost $n$.
	Hence, it is enough to show that in $T$ the maximum cost of an agent is upper bounded by a constant.  
	We clearly have that the cost incurred by a non-leaf agent is strictly smaller than the cost incurred by any of its descendants. Therefore, the maximum costs is achieved by a leaf agent. 
	
	Consider two leafs $u$ and $v$ in $T$ such that $u$ pays the maximum cost. Let $P_v$ be the node-to-root path starting from the parent of $v$. Since $T$ is stable, $cost_{T}(u)< 1 + 1/2 + \sum_{(i,j)\in P_v} \frac{\indeg(j)}{|T(i)|} =1/2+cost_{T}(v)$. 
	Therefore, we only have to show that there exists a leaf agent $v$ with a constant cost value.
	
	We now prove that such a leaf agent always exists. By Corollary~\ref{cor:min_deg_of_child}, each node of in-degree~$d$ has at least $d-1$ children of in-degree at least $\lceil(d-1)/2\rceil$.  
	Consider a root-to-leaf path $P=(r=v_0,\ldots, v_h=v)$ where each next hop goes always towards the smallest appended subtree where the root has an in-degree of at least half of the node's in-degree minus one, i.e., for any $v_{i}\in P$, $v_{i+1}=\text{argmin}{\{|T(w)|:\ (w,v_i)\in E \text{ and }\indeg(w)\geq (\indeg(v_i)-1)/2\}}$. 
	Then for every $0\leq i\leq h-1$,  $|T(v_i)|\geq (\indeg(v_i)-1)|T(v_{i+1})|+2$.
	
	Denote by $|t_k|$ the size of the minimum stable tree with a root of in-degree $k$. Then by Corollary~\ref{cor:min_deg_of_child} and Lemma~\ref{lem:deg_parent_of_leaf} it holds that
	\begin{equation}\label{eq:t_i}
	|t_0|\geq 1,| t_{1}|\geq 2, t_k\geq (k-1)\cdot |t_{\lceil(k-1)/2\rceil}|+2.
	\end{equation}
	We show via induction that for any $k\geq 11$, $|t_k|\geq (2k+1)k^2$. Indeed, it holds that
	$
	|t_{k+1}|\geq k\cdot |t_{\lceil k/2\rceil}|+2 > (k+1)k^3/2^2\geq (2(k+1)+1)(k+1)^2,
	$ where the last inequality holds for all $k\geq 11$.
	
	The overall cost incurred by the leaf $v$ is at most the costs incurred by $v$ for all edges  $(v_i, v_{i-1})\in P$ where in-degree of $v_{i-1}$ is less than the cost incurred by $v$ for all other edges in $P$ plus $2$. 
	
	By Lemma~\ref{thm:same_degree}, each leaf-to-root path has at most three nodes of in-degree $1$, which implies that $v$ pays at most $p_1:=\frac{11}{6}$ for all edges ending in a node with in-degree equals $1$. 
	
	By Lemma~\ref{thm:same_degree}, the in-degrees of the nodes in the leaf-to-root path $P$ strictly increase with at least every second hop. This implies that for $i\leq h-4-(11-1)\cdot 2=h-24$ it is guaranteed that $\indeg(v_i)\geq 11$. Hence, starting from the first node having in-degree at least $11$ in $P$, agent $v$ pays
	\begin{align*}
		p_2:=\sum\limits_{i=1}^{h-24}\frac{\indeg(v_{i-1})}{|T(v_{i})|}
		\leq \sum\limits_{i=1}^{h-24}\frac{2\indeg(v_{i})+1}{|t_{\indeg(v_{i})}|}
		\leq \sum\limits_{i=1}^{h-24}\frac{1}{(\indeg(v_{i}))^2}
		\leq 2\sum\limits_{i=11}^{\infty}\frac{1}{i^2} <2\left(\zeta(2)-\sum\limits_{i=1}^{10}\frac{1}{i^2}\right),
	\end{align*}
	where $\zeta(s)$ is the Riemann zeta function. Hence, $p_2 < 0.2$.
	
	Finally, we need to evaluate the cost of the path $P$ for all nodes $v_i$ with the in-degree $2\leq\indeg(v_i)\leq 10$. 
	Since for every $0\leq i\leq h-1$,  $|T(v_i)|\geq (\indeg(v_i)-1)|T(v_{i+1})|+2$, each edge $(v_{i}, v_{i+1})$ in the path $P$ costs at most $$\frac{2\indeg(v_i)+1}{|T(v_i)|}\leq \frac{2\indeg(v_i)+1}{(\indeg(v_i)-1)|T(v_{i+1})|}=\frac{2}{|T(v_{i+1})|}+\frac{3}{(\indeg(v_i)-1)|T(v_{i+1})|}.$$
	By applying the inequality (\ref{eq:t_i}) and since the in-degrees of the nodes in $P$ increase at most with every second level, it holds that the total cost of the subpath is at most $p_3:=2\sum\limits_{i=2}^{9}\frac{2}{t_i}+\frac{2}{t_1}+\frac{2}{t_{10}}+\sum\limits_{i=2}^{10}\left(\frac{3}{(i-1)t_{i-1}}+\frac{3}{(i-1)t_i}\right)< 3.12 + 2.975<6.01$.
	Therefore, the total cost of the path $P$ payed by an agent $v$ is strictly less than $p_1+p_2+p_3<8.12$. This implies that the PoA is at most 8.62. \qedhere
\end{proof}

\noindent We now prove a lower bound to the PoA using the extremal stable balanced trees of Theorem~\ref{thm:some_balanced_trees_are_NE}. (See Figure~\ref{fig:9_4_2_1_apx}.) For the rest of this section, let ${\mathcal T}_h$ denote the extremal balanced tree of height $h \geq 1$ and degree sequence $\dd_h=0$, $\dd_{h-1}=1$, $\dd_{h-2}=2$ (if $h\geq 2$), $\dd_{h-3}=4$ (if $h \geq 3$), and $\dd_{i}=2\dd_{i+1}+1$ for every $i \leq h-4$. We will denote by $sc_h$ and $n_h$ the social cost and the number of nodes (root included) of ${\mathcal T}_h$.

\begin{theorem}\label{thm:lower_bound_PoA}
	The PoA is at least $2.4317$.
\end{theorem}

\begin{proof}
	From Theorem~\ref{thm:opt_is_a_path}, we know that the social cost of a social optimum is equal to $n$. We prove the claimed lower bound by showing that, for every $h \geq 7$, the social cost of ${\mathcal T}_h$ is at least $2.4317 (n_h-1)$. By Theorem~\ref{thm:some_balanced_trees_are_NE}, we have that ${\mathcal T}_h$ is stable for every $h$. 
	
	First, we compute the exact values of $sc_7$ and $n_7$. Since $|T_{i-1}|= \dd_{i-1}|T_i|+1$ and $|T_7|=1$, we have that $|T_6|=2$, $|T_5|=5$, $|T_4|=21$, $|T_3|=190$, $|T_2|=3611$, $|T_1|= 140830$, and $|T_0|=11125571$. Hence, $n_7=11125571$. Since $sc_0=0$ and $sc_{i-1}=\dd_{h-i+1}s_i+\dd_{h-i+1}^2$, we have that $sc_1=1$, $sc_2=6$, $sc_3=40$, $sc_4=441$, $sc_5=8740$, $sc_6=342381$, and $sc_7=27054340$. Therefore, \[
	sc_7 = 27054340 > 2.4317 \cdot 11125570 = 2.4317 (n_7-1).
	\]
	
	\noindent Next, we prove that $sc_h > 2.4317 (n_h-1)$ for every $h \geq 8$. Let $z$ be the number of nodes of ${\mathcal T}_h$ that are at distance $h-7$ from the root $r$. It holds that $z$ is equal to the number of trees ${\mathcal T}_7$ that are contained in ${\mathcal T}_h$. Furthermore, the overall sum of the number of nodes of ${\mathcal T}_h$ that are at distance of at most $h-7$ from $r$ is upper bounded by $2z$. This implies that $sc_h \geq sc_7 \cdot z = 27054340 \cdot z$, while $n_h \leq n_7 \cdot z + 2z = 11125571 \cdot z + 2z= 11125573 \cdot z$. Therefore,  
	\[
	sc_h \geq 27054340 \cdot z > 2.4317 \cdot 11125573 \cdot z \geq 2.4317 \cdot n_h > 2.4317 (n_h-1).
	\]
	This completes the proof.
\end{proof}

\noindent Next, we prove an upper bound to the average agent's cost in ${\mathcal T}_h$ and provide an interesting conjecture. We define $a_h\coloneqq sc_h/(n_h-1)$ as the average agent's cost in ${\mathcal T}_h$.

\begin{lemma}\label{lem:avg_cost_in_BT}
	For every $h \geq 1$, $a_h \leq 2.4318$.
\end{lemma}

\begin{proof}
	From the proof of Theorem~\ref{thm:lower_bound_PoA} we have that 
	\begin{align*}
		a_1=&1/(2-1)=1 \\
		a_2=&6/(5-1)=1.5\\
		a_3=&40/(21-1)=2\\
		a_4=&441/(190-1)\leq 2.34\\
		a_5=&8740/(3611-1)\leq 2.43\\
		a_6=&342381/140830 \leq 2.4312\\
		a_7=&27054340/(11125571-1)\leq 2.43173.
	\end{align*} We now prove by induction that $a_{h+1} \leq a_h+0.00005 \cdot \frac{1}{2^{h-7}}=a_h + \frac{1}{20,000\cdot 2^{h-6}}$ for every $h \geq 7$. Thus, showing that $a_h \leq 2.4318$ for every $h$. Since $\dd_0 \leq 3\dd_1$ and $n_h > \dd_1 n_{h-1}$, we have that 
	\[
	a_{h+1} \leq \frac{sc_{h+1}}{n_{h+1}-1} = \frac{\dd_0 sc_h + \dd_0^2}{\dd_0 n_h+1-1}=\frac{sc_h+\dd_0}{n_h} \leq a_h+\frac{\dd_0}{n_h} < a_h+\frac{3\dd_1}{\dd_1 n_{h-1}}\leq a_h + \frac{3}{n_{h-1}}.
	\] 
	We now complete the proof by showing via induction on $h$ that $\frac{3}{n_{h-1}}\leq \frac{1}{20000 \cdot 2^{h-6}}$ for every $h \geq 7$. For the base case $h=7$, $n_6=140830$ and therefore, $\frac{3}{140830} \leq \frac{1}{20000 \cdot 2}$. Now, if we assume that $\frac{3}{n_{h-1}}\leq \frac{1}{20000 \cdot 2^{h-6}}$, since $n_h > \dd_1 n_{h-1}$ and $\dd_1 \geq 2$, we have that
	\[
	\frac{3}{n_h} < \frac{3}{\dd_1 n_{h-1}} \leq \frac{1}{\dd_1 \cdot 20000 \cdot 2^{h-6}} < \frac{1}{20000 \cdot 2^{h+1-6}}.
	\]
	This completes the proof.
\end{proof}

\subsection{Price of Stability}\label{sec:PoS}
\noindent We now turn our focus to the PoS and prove a lower bound.
\begin{theorem}\label{thm:LB_PoS}
	The PoS is at least $\frac{7}{5}-\varepsilon$, for $\varepsilon\in \Theta(1/n)$.
\end{theorem}

\begin{proof}
	We know that the social cost of a tree $T$ is the sum of squared in-degrees of all nodes including the root. Consider the following procedure: for each node $v$ in $T$ which has in-degree larger than~$2$, swap one of its children  $u$ to a node $v'$ of in-degree 1 or 0 closest to the root such that it does not disconnect the tree. The resulting tree $T'$ has  all nodes of in-degree at most 2, and social cost $SC(T')\leq SC(T)$. Indeed, each step changes the social cost by $(\indeg(v)-1)^2-\indeg(v)^2 + (\indeg(v')+1)^2 - \indeg(v')^2=2\indeg(v')-2\indeg(v)+2<0$, since $\indeg(v')\leq 1$, i.e., the social costs decrease. Let $h$ be the height of the maximal subtree in $T'$ such that all its nodes are of degree 2. We can assume that $T'$ is as much balanced as possible, i.e., there are no nodes with the same in-degree which differ by more than one level since otherwise we can swap nodes of the higher level to the nodes of the lower level such the number of nodes of each in-degree remains the same.
	
	The number of nodes with in-degree equals $2$ in $T'$ is $\sum\limits_{i=0}^{h-1}2^i + k_2$, where $k\geq 0$ is the number of nodes with in-degree equals $2$ at distance $h+1$ from the root. Denote $k_1$ and $k_0$ as the number of nodes with in-degree equals $1$ and leaf nodes, respectively, in~$T'$. The procedure above does not create new nodes of in-degree equals $1$ or leaf nodes in a tree. Also, by Lemma~\ref{thm:same_degree}, any sequence of nodes with in-degree equals $1$ starting from a leaf node in $T'$ contains at most $4$ nodes. It implies that $k_1\leq 3k_0\leq 3\cdot 2^h$.  
	Since the social cost of the optimal network is equal to the number of nodes, cf. the proof of Theorem~\ref{thm:PoA}, the PoS is lower bounded by
	\begin{align*}
		PoS&\geq\frac{SC(T')}{n}=\frac{2^2\cdot \left(\sum\limits_{i=0}^{h-1}2^i + k_2\right) + k_1}{\sum\limits_{i=0}^{h-1}2^i+ k_2+k_1+k_0}\geq \frac{2^2\cdot\sum\limits_{i=0}^{h-1}2^i + k_1}{\sum\limits_{i=0}^{h-1}2^i + k_1+k_0}\\
		&\geq \frac{2^2\left(2^h-1\right)+k_1}{2^h-1+k_1+k_0}\geq \frac{2^{h+2}+k_1-4}{2^h-1+k_1+2^h}\geq \frac{2^{h+2}+3\cdot 2^h-4}{2^{h+1}+3\cdot 2^h-1}>\frac{7}{5}-\varepsilon,
	\end{align*}
	where $\varepsilon\in \Theta(2^{-h})=\Theta(n^{-1})$.
\end{proof}

\noindent Next, we investigate the PoS in certain balanced trees and prove an upper bound which is strictly better than the upper bound on the PoA.
\begin{theorem}\label{thm:PoS_BT_UB}
	For all $n\in\mathbb{N}$ such that there is a balanced tree $T$ of size $n$ with the in-degree sequence  $(0,1,2,4,\dd_{h-4},\ldots,\dd_0)$, where $\dd_i\leq 2\dd_{i+1}+1$ for $i\leq h-4$, the PoS is at most 2.83.
\end{theorem}

\begin{proof}
	For those values $n$ such that there exist a balanced tree $T$ of size $n$ with the in-degree sequence  $(0,1,2,4,\dd_{h-4},\ldots,\dd_0)$, where $\dd_i\leq 2\dd_{i+1}+1$ for $i\leq h-4$, we can provide an upper bound for the PoS, which is strictly better than the upper bound for the PoA. Clearly, for all other values of~$n$ where there exist an equilibrium, the PoS is at most the PoA value, which is at most $8.62$.
	
	Consider a balanced tree $BT_h(\dd_0)$ rooted at a node of in-degree equal to $\dd_0$ of height $h$. Its social cost  $SC(BT_h(\dd_0))=\sum\limits_{i=0}^{h-1}\left(\dd_{i}^2\cdot\prod\limits_{j=0}^{i-1}\dd_{j}\right)=SC(BT_{h-1}(\dd_{1}))\cdot \dd_0 + \dd_0^2$, while the social cost of the corresponding \OPT, i.e., a path graph $P$, is $SC(P)=|T_h(\dd_0)|=\sum\limits_{i=0}^{h-1}\dd_{i}= |BT_{h-1}(\dd_{1})|\cdot \dd_0 + \dd_0$. It holds that
	
	\begin{align*}
		PoS &=\frac{SC(BT_h(\dd_0))}{|BT_h(\dd_0)|}=\frac{SC(BT_{h-1}(\dd_{1}))+\dd_0}{|BT_{h-1}(\dd_{1})|+1}<\frac{SC(BT_{h-1}(\dd_{1}))}{|BT_{h-1}(\dd_{1})|}+\frac{\dd_0}{|BT_{h-1}(\dd_{1})|}\\
		&<\ldots<\frac{SC(BT_3(\dd_{h-3}))}{|BT_3(\dd_{h-3})|}+\sum\limits_{i=0}^{h-4}\frac{\dd_{i}}{|BT_{h-i-1}(\dd_{i+1})|} \leq 2+\sum\limits_{i=0}^{h-4}\frac{2\dd_{i+1}+1}{|BT_{h-i-1}(\dd_{i+1})|}\\
		&= 2+\sum\limits_{i=0}^{h-4}\frac{2}{|BT_{h-i-2}(\dd_{i+2})|+1}+\sum\limits_{i=0}^{h-4}\frac{1}{|BT_{h-i-1}(\dd_{i+1})|}\\
		&=2+\sum\limits_{i=2}^{h-3}\frac{3}{|BT_{h-i}(\dd_{i})|}+\frac{1}{|BT_{h-1}(\dd_{1})|}+\frac{2}{|BT_{2}(\dd_{h-2})|}\\
		&=2.4+\sum\limits_{i=2}^{h-3}\frac{3}{|BT_{h-i}(\dd_{i})|}+\frac{1}{|BT_{h-1}(\dd_{1})|}.
	\end{align*}
	
	\noindent The size of the balanced tree with the root of in-degree $\dd_0$ of height $h$ is equal to $|BT_h(\dd_0)|=\sum\limits_{i=1}^{h-1}\prod\limits_{j=0}^i \dd_j>\prod\limits_{j=0}^{h-1} \dd_{j} > 2^{h(h-1)/2}$. Then $|BT_{h-i}(\dd_{i})|>2^{(h-i)(h-i-1)/2}$ and we get
	\begin{align*}
		PoS < 2.4+\sum\limits_{i=2}^{h-3}\frac{3}{2^{(h-i)(h-i-1)/2}}+\frac{1}{|BT_3(2)|} < 2.4+0.43+0.2=2.83. \hspace*{5cm} \qedhere
	\end{align*}
\end{proof}

\subsection{Fairness measure}\label{sec:fairness}
We investigate the Fairness Ratio which considers the cost distribution among the agents. The $\FR$ in the social optimum turns out to be equal to $nH_n$, whereas the $\FR$ in stable trees is $o(n)$ and hence stable trees admit a more fair cost-sharing. 
%We start with the statement about the social optimum.

\begin{theorem}\label{thm:FR_OPT}
The Fairness Ratio for $\OPT_n$ is  $nH_n$, where $H_n = \sum\limits_{i=1}^n\frac{1}{i}$ is the $n$-th harmonic number.
\end{theorem}

\begin{proof}
Consider a  path $P = u_0, u_1, \ldots, u_n$ where $r = u_0$, and $u_i$ is a node at distance $n-i$ from the root. 
Clearly $\text{cost}_P(u_{i+1}) \geq \text{cost}_P(u_{i})$ for $i \in \{1, \ldots, n-1\}$. 
Hence, for the Fairness Ratio on the path we need to consider the costs incurred by $u_0$ and $u_{n}$.
\[\FR(P)=\frac{cost_P(u_n)}{cost_P(u_1)}=\frac{\sum_{i=1}^n\frac{1}{i}}{1/n} = n\cdot H_n\leq n(\ln n+1)\qedhere\]
\end{proof}

\noindent We now turn our focus to the analysis of the class of all stable trees. Based on the lower and upper bounds for the in-degree of the root, we prove that $\FR$ is in $o(n)$. 

\begin{theorem}\label{FR_UB_general}
	The Fairness Ratio for any NE tree is at most $$\frac{8.62(n-2)\cdot \ln\ln(4\sqrt{n/5})}{\ln(4\sqrt{n/5})} - 2^{13}\cdot\left(1-\frac{2\ln\ln(4\sqrt{n/5})}{\ln(4\sqrt{n/5})}\right)\cdot \left(\frac{\ln(4\sqrt{n/5})}{\ln\ln(4\sqrt{n/5})}\right)^{\log\left(\sqrt{\frac{\ln(4\sqrt{n/5})}{\ln\ln(4\sqrt{n/5})}}\right)-5.5},$$ which is at most $8.62\cdot\frac{(n-2)\cdot \ln\ln(4\sqrt{n/5})}{\ln(4\sqrt{n/5})}$.
\end{theorem}

\begin{proof}
As shown in the proof of Theorem~\ref{thm:PoA}, the cost of any agent in a stable tree is upper bounded by $8.62$. 
Clearly the minimal cost in a stable tree is paid by a node adjacent to the root. By Theorem~\ref{thm:LB_deg_of_the_root}, the in-degree of the root $d_0$ is at least $\frac{\ln(4\sqrt{n/5})}{\ln\ln(4\sqrt{n/5})}$, 
while, by Corollary~\ref{cor:min_deg_of_child},  the maximum size of a tree adjacent to the root is at most $n-(d_0-2)t_{\min}-2$, where $t_{\min}$ is the minimum size of the tree rooted at a child of the root of in-degree at least $(d_0-1)/2$.

We will evaluate $t_{min}$. By Lemma~\ref{thm:same_degree}, $\indeg(v_i)\geq 2$ for any node $v_i$ at distance $i$ from the root, for $i\leq h-4$. 
Then by Corollary~\ref{cor:min_deg_of_child}, the size of the subtree $T(v_i)$ with the root at node $v_i$ is at least $\sum\limits_{\ell=0}^{k-1}\left(x_\ell\prod\limits_{j=0}^{\ell-1}(x_j-1)\right)$ for $k=\lfloor \log((\indeg(v_i)+1)/3)\rfloor$, where the sequence $\{x_\ell\}_{\ell=0}^k$ such that $x_\ell\geq (x_{\ell+1}-1)/2$ and $x_0 \coloneqq v_i$. 
Then, from the proof of Theorem~\ref{thm:UB_deg_of_the_root}, $|T(v_i)| > 2^{\sum_{j=0}^{k-3}(k-j-2)}=2^{\frac{(k-1)(k-2)}{2}}$.
Since $\indeg(v_1)\geq (d_0-1)/2$ and thus $k > \log((d_0+1)/6)$, we have:
\begin{align*}
t_{\min} &> 2^{\frac{(\log((d_0+1)/6)-1)(\log((d_0+1)/6)-2)}{2}} 
> \left(2^{(\log(d_0+1)-4)}\right)^{(\log(d_0+1)-5)/2}=\left((d_0+1)\cdot 2^{-4}\right)^{(\log(d_0+1)-5)/2}\\
&= (d_0+1)^{\log(d_0+1)/2}\cdot(d_0+1)^{-5/2}\cdot 2^{-2(\log(d_0+1)-5)}
=(d_0+1)^{\log(d_0+1)/2-5/2-2}\cdot 2^{10} \\
& > 2^{10}\cdot d_0^{\log(d_0)/2-4.5} 
\end{align*}

\noindent Hence, the $\FR$ for any stable tree is at most 
\begin{align*}
&\frac{8.62\left(n-2-2^{10}(d_0-2) d_0^{\log({\sqrt{d_0}})-4.5}\right)}{d_0}
=\frac{8.62(n-2) \ln\ln(4\sqrt{n/5})}{\ln(4\sqrt{n/5})} - 8.62\left(1-\frac{2}{d_0}\right)\cdot d_0^{\log({\sqrt{d_0}})-5.5}\\
&< \frac{8.62(n-2)\cdot \ln\ln(4\sqrt{n/5})}{\ln(4\sqrt{n/5})} - 2^{13}\cdot\left(1-\frac{2\ln\ln(4\sqrt{n/5})}{\ln(4\sqrt{n/5})}\right)\cdot \left(\frac{\ln(4\sqrt{n/5})}{\ln\ln(4\sqrt{n/5})}\right)^{\log\left(\sqrt{\frac{\ln(4\sqrt{n/5})}{\ln\ln(4\sqrt{n/5})}}\right)-5.5}\\
&< \frac{8.62(n-2)\cdot \ln\ln(4\sqrt{n/5})}{\ln(4\sqrt{n/5})}.\qedhere
\end{align*}

\end{proof}

\begin{theorem}\label{thm:FR_LB}
The Fairness Ratio for a stable tree is at least $n\cdot 2^{-2\sqrt{2\log(n)}}$.
\end{theorem}

\begin{proof}
By Theorem~\ref{thm:PoA}, the maximum cost incurred by any agent is the cost incurred by a leaf, while, by Lemma~\ref{lem:deg_parent_of_leaf}, any leaf is a child of a node with in-degree equals $1$, and thus, any leaf pays at least $1$.

The minimum cost in a stable tree is incurred by a node which is a child of the root. Hence, the minimum cost is equal to $\frac{d_0}{t_{\max}}$, where $t_{\max}$ is the size of the largest subtree rooted at a child of the root. Clearly, $t_{\max}\geq\frac{n}{d_0}$. Thus, $\frac{d_0}{t_{\max}}\leq \frac{d_0^2}{n}$. 
As a consequence, we get that $\FR \geq \frac{n}{d_0^2}$, which is at least $\frac{n}{2^{2\sqrt{2\log(n)}}}$ by Theorem~\ref{thm:UB_deg_of_the_root}.
\end{proof}

\noindent Finally, we investigate the class of stable balanced trees and prove a more precise upper bound.
\begin{theorem}\label{FR:BT}
	The Fairness Ratio for a stable balanced tree with the in-degree sequence\\ 
	$(0,1,2,4,\dd_{h-4},\ldots,\dd_0)$, where $\dd_i\leq 2\dd_{i+1}+1$ for $i\leq h-4$,  is at most $\frac{2.4318n\cdot \left(\ln\ln(4\sqrt{n/5})\right)^2}{\left(\ln(4\sqrt{n/5})\right)^2}$.
\end{theorem}

\begin{proof}
By Lemma~\ref{lem:avg_cost_in_BT}, the cost of any agent in a balanced stable tree is at most 2.4318. 
By Theorem~\ref{thm:LB_deg_of_the_root}, the in-degree of the root is at least $\frac{\ln(4\sqrt{n/5})}{\ln\ln(4\sqrt{n/5})}$, while the size of a tree adjacent to the root $r$ is $n/\indeg(r)$.
Then the \FR for any stable balanced tree is at most $\frac{2.4318n\cdot \left(\ln\ln(4\sqrt{n/5})\right)^2}{\left(\ln(4\sqrt{n/5})\right)^2}$.
\end{proof}

\section{Extensions for Future Work: The Path Version and Coalitions}\label{sec:extension}
A natural extension of our model is to allow for a richer strategy space. Instead of selecting a single outgoing edge, agents could strategically select a complete path towards the root $r$. This version, called the \emph{path-\SNCG}, is closer to the fair single-source connection game by Anshelevich et al.~\cite{ADTW08,ADKTWR}. See Appendix~\ref{apx_path_version} for a formal definition of the path-\SNCG.

We give some preliminary results relating the equilibria of the \SNCG to the equilibria of the path-\SNCG. Our results indicate that studying the path-\SNCG, in particular its PoA and PoS, is a promising next step.
We start with showing that also in the path-\SNCG all equilibria must be trees.

\begin{lemma}\label{lemma:path_ne_trees}
Any equilibrium network in the path-\SNCG is a tree. 
\end{lemma}

\begin{proof}
	We prove the statement via contradiction. Assume there is an agent $a$ who has two paths $P$ and $P^\prime$ to the root $r$. If there are more than one such agent, let $a$ be the last agent in $P$ which also belongs to $P^\prime$, thus, the one closest to the root. Without loss of generality, let $\cost(P) \leq \cost(P^\prime)$. Let~$a^\prime$ be an agent who chooses $P^\prime$ and let $\cost^\prime(P)$ be the cost of $P$ if $a^\prime$ picks path $P$ instead of path~$P^\prime$. It holds that $\cost^\prime(P) < \cost(P)$ since there is an additional agents who pays for the edges in~$P$. Hence, replacing $P^\prime$ by $P$ is an improvement for $a^\prime$ and therefore every agent has a unique path to the root~$r$.
\end{proof}

\noindent Now we show that the \SNCG can be considered as a special case of the path-\SNCG since all equilibrium trees of the \SNCG are equilibria in the path-\SNCG but not vice versa. 
 
\begin{theorem}\label{thm:path_version_eq_neighb_version_comm_sink}
 The set of NE in the path-\SNCG is a superset of the set of NE in the \SNCG.
 \end{theorem}

\begin{proof}
	To prove the claim, we show that any NE in the \SNCG is a NE in the path-\SNCG. Then we provide an example of a NE in the path-\SNCG that is not in equilibrium for the \SNCG.
	
	Consider an arbitrary NE $T_N$ of the \SNCG and assume towards a contradiction that there is an agent~$a$ who still has an improvement in $T_N$ in the path-\SNCG. Thus, agent $a$ switches to a better path~$P'$. In this case, agent $a$ has to pay at least one edge, say edge $e'=(u,v)$, of $P'$ by herself. This is true because of Lemma \ref{lemma:path_ne_trees}. If there is more than one such edge, then let $(u,v)$ be the edge which is closest to the root $r$. In that case, switching agent $a$'s edge to $v$ is also an improvement for agent $a$ in the \SNCG, since agent $a$ pays $indeg(v)$ for $(u,v)$ in the path-\SNCG, whereas in the \SNCG she will pay at most $\frac{indeg(v)}{|T(a)|} \leq indeg(v)$ for the new edge $(a,v)$. 
	
	To show that there exist NE in the path-\SNCG which are not in NE for the \SNCG, consider the tree shown in Figure~\ref{fig:pathEQ_samples}~(left). Let $r$ be the common sink for all agents. Since agent $g$ can swap her edge towards $h$ and decrease her costs from $\frac{14}{9}$ to $\frac{50}{33}$ the depicted tree is not in equilibrium for the \SNCG.
	
	To show that the tree in Figure~\ref{fig:pathEQ_samples}~(middle) is indeed in equilibrium in the path-\SNCG, we will show that no agent can unilaterally improve her strategy. Due to Lemma \ref{lemma:path_ne_trees}, when an agent $x$ switches to another path $P^\prime$, $x$ has to pay for at least one edge by herself. Note that due to symmetry, many agents can be treated equally. 
	\begin{itemize}
		\item Agent $a$ has costs of $\frac29$. Since any deviation from the current strategy would lead to costs larger than $1$, since agent $a$ needs to pay for at least one edge by herself, no improvement is possible. The same is true for agent $j$.
		\item Agent $b$ has costs of $\frac{13}{18}$. Again, any deviation from the current strategy leads to costs larger than~$1$, hence, no improvement is possible. The same holds for agents $c$, $h$ and $k$.
		\item Agent $d$ has costs of $\frac{19}{18}$. Since the current costs are smaller than $2$, agent $d$ can only pay for a single edge with costs $1$ by herself. This is only true for an edge to the leaf nodes $h$, $i$, $p$ and $q$. However, any deviation from the current strategy which contains an edge to a leaf node has costs larger than $\frac{11}{6}$ and hence, no improvement is possible. The same holds for the agents $e$, $l$ and $m$.
		\item Agent $f$ has costs of $\frac{14}{9}$. Since the current costs are smaller than $2$, agent $d$ can only pay for a single edge with costs $1$ by herself. Again, this is only true for an edge to one of the leaf nodes. Any deviation from the current strategy which contains an edge to a leaf node has costs larger than $\frac{11}{6}$ and hence, no improvement is possible. The same holds for the agents $g$, $n$ and $o$.
		\item Agent $h$ has costs of $\frac{23}{9}$. Since the current costs are smaller than $3$, agent $h$ can only pay for edges with costs of maximum $2$ by herself. Choosing a path over a node with current degree $1$ costs $2$. Hence, agent $h$ is not able to pay for another edge by herself. However, this leads to costs of at least $\frac35$ and therefore this is not an improving strategy change. In addition, any deviation from the current strategy which contains an edge to a leaf node has costs of at least $\frac{161}{60}$ and hence, no improvement is possible. The same is true for the agents $i$, $p$ and $q$. \qedhere
	\end{itemize}	
\end{proof}

	\begin{figure}
	\centering
	\begin{subfigure}{0.3\textwidth}
		\centering
		\includegraphics[width=0.4\linewidth]{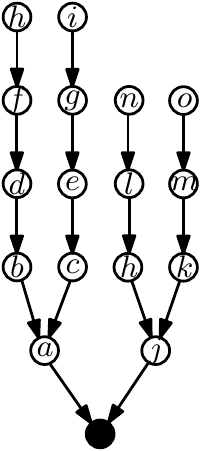}
	\end{subfigure}
	\begin{subfigure}{0.3\textwidth}
		\centering
		\includegraphics[width=0.4\linewidth]{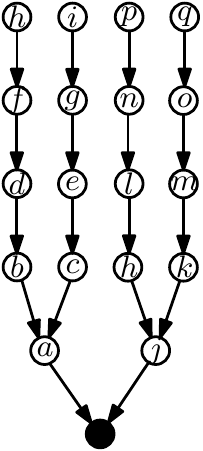}
	\end{subfigure}
	\begin{subfigure}{0.3\textwidth}
		\centering
		\includegraphics[width=0.6\linewidth]{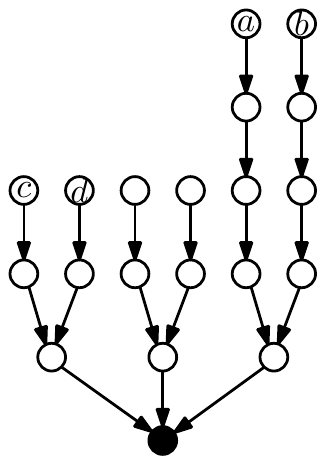}
	\end{subfigure}
	\caption{Left: A path-\SNCG NE that is not a \SNCG NE for $n = 16$. Middle: A path-\SNCG NE that is not a \SNCG NE for $n = 18$. Right: A \SNCG NE that is not a strong path-\SNCG NE. }
	\label{fig:pathEQ_samples} 
\end{figure}

\noindent We showed for the \SNCG that for $n = 16$ and $n = 18$ there exists no stable network. We contrast this negative result with a NE existence proof for the path-\SNCG for the corresponding values. Figure~\ref{fig:pathEQ_samples} (left and middle) show equilibrium trees for the path-\SNCG for $n = 16$ and $n = 18$, respectively.

\begin{theorem}\label{path_existence}
For $n = 16$ and $n = 18$ there exists a stable network for the path-\SNCG.
\end{theorem}

\begin{proof}
	In the proof of Lemma~\ref{thm:path_version_eq_neighb_version_comm_sink} we already showed that there exists an equilibrium for the path-\SNCG for $n = 18$. Hence, it remains to show that there also exists an equilibrium for $n = 16$. Consider the tree shown in Figure~\ref{fig:pathEQ_samples}~(left). We will show that no agent can unilaterally improve her strategy. Due to Lemma \ref{lemma:path_ne_trees}, when an agent $x$ switches to another path $P^\prime$, $x$ has to pay for at least one edge by herself. Note that due to symmetry, many agents can be treated equally. 
	\begin{itemize}
		\item Agent $a$ has costs of $\frac29$. Since any deviation from the current strategy would lead to costs larger than $1$, since agent $a$ needs to pay for at least one edge by herself, no improvement is possible. The same is true for agents $b$ and $c$ who have costs of $\frac{13}{18}$, agent $j$ who has costs of $\frac27$ and agents $h$ and $k$ who have costs of $\frac{20}{21}$.
		\item Agent $d$ has costs of $\frac{19}{18}$. Since the current costs are smaller than $2$, agent $d$ can only pay for a single edge with costs $1$ by herself. This is only true for an edge to the leaf nodes $h$, $i$, $n$ or $o$. However, any deviation from the current strategy which contains an edge to a leaf node hast costs larger than $\frac32$. Hence, no improvement is possible. The same holds for agent $e$.
		\item Agent $f$ has costs of $\frac{14}{9}$. Since the current costs are smaller than $2$, agent $f$ can only pay for a single edge with costs $1$ by herself. This is only true for an edge to the leaf nodes $h$, $i$, $n$ or $o$. However, any deviation from the current strategy which contains an edge to a leaf node hast costs larger than $\frac{11}{6}$. Hence, no improvement is possible. The same holds for agent $g$.
		\item Agent $h$ has costs of $\frac{23}{9}$. Since the current costs are smaller than $3$, agent $h$ can only pay for edges with costs of maximum $2$ by herself. Choosing a path over a node with current degree $1$ costs $2$. Hence, agent $h$ is not able to pay for another edge by herself. However, this leads to costs of at least $\frac{118}{45}$ and therefore this is not an improving strategy change. In addition, any deviation from the current strategy which contains an edge to a leaf node hast costs of at least $\frac{31}{12}$ and hence, no improvement is possible. The same is true for agent~$i$.
		\item Agent $l$ has costs of $\frac{61}{42}$. Since the current costs are smaller than $2$, agent $l$ can only pay for a single edge with costs $1$ by herself. This is only true for an edge to the leaf nodes $h$, $i$, $n$ or $o$. However, any deviation from the current strategy which contains an edge to a leaf node hast costs larger than $\frac{55}{21}$. Hence, no improvement is possible. The same holds for agent $m$.
		\item Agent $n$ has costs of $\frac{103}{42}$. Since the current costs are smaller than $3$, agent $n$ can only pay for edges with costs of maximum $2$ by herself. Choosing a path over a node with current degree $1$ costs $2$. Hence, agent $n$ is not able to pay for another edge by herself.  However, this leads to costs of at least $\frac{13}{5}$ and therefore this is not an improving strategy change. In addition, any deviation from the current strategy which contains an edge to a leaf node hast costs of at least $\frac{55}{21}$ and hence, no improvement is possible. The same is true for agent~$o$. \qedhere
	\end{itemize}
\end{proof}

\noindent Together with Theorem~\ref{path_existence} and since any NE in the \SNCG is a NE in the path-\SNCG, we go along with Conjecture~\ref{conj:existence} and believe that for all values of $n$ stable trees exist for the path-\SNCG.

\begin{conjecture}
For any $n \in \mathbb{N}$ a pure NE exists in the path-\SNCG.
\end{conjecture}

\noindent An agent $a$ in the \SNCG benefits from the fact that if $a$ changes her strategy and switches her edge towards another node the costs of the new edge is also shared among all of $a$'s ancestors.
It seems natural to consider a strategy change in the \SNCG as a coalitional strategy change in the path-\SNCG by the coalition consisting of agent $a$ and all her ancestors. So NE in the \SNCG could be in strong NE~\cite{AFM09} for the path\SNCG. However, we show that this is not true, see Figure~\ref{fig:pathEQ_samples} (right).

\begin{theorem}\label{thm_neighEQ_not_strongpathEQ}
There is a NE in the \SNCG which is not in strong NE for the path-\SNCG. 
\end{theorem}

\begin{proof}
	Consider the tree depicted in Figure~\ref{fig:pathEQ_samples}~(right). We have already seen that the tree is in NE for the \SNCG, cf. Figure~\ref{fig:sample_trees}. However, it is not in strong NE for the path-\SNCG, since agents $a$ and $b$ can form a coalition and jointly change their strategy, such that agent $a$ chooses the path over node $c$ and agent $b$ over node $d$. With this agents $a$ and $b$ can both decrease their cost from $\frac83$ to $\frac{109}{42}$.
\end{proof}

\section{Conclusion}
We have studied a tree formation game to investigate how selfish agents self-organize to connect to a common target in the presence of dynamic edge costs that are sensitive to node degrees. This mimics settings in which nodes can charge prices for offering their routing service and where these prices are guided by supply and demand, i.e., more popular nodes with higher in-degree can charge higher prices to make up for their increased internal coordination cost.

Our main findings are that our game admits equilibrium trees with intriguing properties like low height, low maximum degree, almost optimal cost, and a somewhat fair distribution of the total cost among the agents. The set of equilibrium trees seems to be combinatorially rich, and characterizing stable trees that are not balanced seems an exciting and challenging problem for future research. It would also be interesting to study the degree distribution in stable trees and to evaluate possible connections with power-law degree distributions which are ubiquitous in real-world networks.    

We note in passing that our model can easily be generalized to settings with more than one target node as long as every possible incident edge may be activated. In this case, several disjoint trees, one for each target node, will be formed. Things change if target nodes and agent nodes may be co-located, and exploring this variant might be interesting.  

\section*{Acknowledgment}
We thank Warut Suksompong for many interesting discussions. Moreover, we are grateful to our anonymous reviewers for their valuable suggestions. This work has been partly supported by COST Action CA16228 European Network for Game Theory (GAMENET).

%\newpage

\bibliographystyle{abbrv}
\bibliography{sourcesinkcreation}

\appendix

\section{Formal Definition of the Path Version}\label{apx_path_version}
The path version of the \SNCG, called path-\SNCG, is defined by a complete directed \emph{host-graph} $H = (V,E)$ with $n$ nodes and $k$ target sink pairs $(s_1,t_1),\dots,(s_k,t_k)$, where pair $(s_i,t_i)$ models that a selfish agent $i$ wants to connect $s_i \in V$ and $t_i \in V$. For this, each agent $i$ strategically selects a path $P_i \subseteq E$ which connects $s_i$ and $t_i$. Note that we treat paths simply as sets of edges. 
The $k$-dimensional vector of all chosen paths $\mathbf{p}=(P_1,\dots,P_k)$ then induces the subgraph $G(\mathbf{p})$ of $H$, which is defined as follows: $G(\mathbf{p}) = (V(\mathbf{p}),E(\mathbf{p}))$, where 
$$V(\mathbf{p}) = \bigcup_{i=1}^k \big\{u \mid \exists (u,v) \in P_i\big\} \quad \text{ and } \quad E(\mathbf{p}) = \bigcup_{i=1}^k P_i,$$ that is, the induced subgraph $G(\mathbf{p})$ consists of all edges which are contained in at least one strategy and the corresponding incident nodes of $H$. We call $G(\mathbf{p})$ the \emph{created network} by strategy vector $\mathbf{p}$.

The cost of agent $i$ depends on the structure of the created network $G(\mathbf{p})$. 
We assume that any edge in a strategically chosen path has a price which depends on the in-degree of its ancestor in $G(\mathbf{p})$ and on the number of other agents using that edge. 
Let $U(u,v)$ denote the set of users of edge $(u,v)$ in $G(\mathbf{p})$ which is defined as follows: $U(u,v) = \{i \mid (u,v) \in P_i\}$. 
Then, the cost of agent $i$ in $G(\mathbf{p})$ is 
\begin{equation*}
	\cost_{G(\mathbf{p})}(i) := \begin{cases}
		\sum_{(u,v)\in P_i}\frac{indeg_{G(\mathbf{p})}(v)}{|U(u,v)|},	&	\text{if $s_i$ and $t_i$ are connected,}\\
		\infty,													&	\text{otherwise.}
	\end{cases}
\end{equation*}
The cost function can be interpreted as the total cost of all edges in a path chosen by an agent, 
where the price of each edge is proportional to a load of its endpoint, i.e., its in-degree, fairly shared among users of the edge.

The \emph{social cost} of strategy vector $\mathbf{p}$, $SC(G(\mathbf{p}))$ for short, is simply the sum of the costs of all agents. That is, $SC(G(\mathbf{p})) = \sum_{i=1}^k \cost_{G(\mathbf{p})}(i)$.

We say that strategy vector $\mathbf{p}$ is in \emph{pure Nash equilibrium} if no agent can unilaterally change her strategy and thereby strictly improve her costs. For a pure Nash equilibrium $\mathbf{p}$ we call the corresponding created network $G(\mathbf{p})$ \emph{stable}.

\end{document}